\documentclass[journal]{IEEEtran}

\usepackage[utf8]{inputenc}
\usepackage[T1]{fontenc}
\usepackage{url}
\usepackage{cite}

\usepackage[cmex10]{amsmath}
\usepackage{amsthm}
\usepackage{graphicx}
\usepackage{amsfonts}
\usepackage{amssymb}
\usepackage{mathrsfs}
\usepackage{mdwmath}
\usepackage{mdwtab}
\usepackage{dsfont}
\usepackage{tcolorbox}
\usepackage{bm}
\usepackage{algorithm}
\usepackage{algorithmic}
\usepackage{url}
\usepackage{pdfsync}
\usepackage{subcaption}
\usepackage{multirow}
\usepackage{array}
\usepackage{booktabs}
\newcolumntype{T}[1]{>{\centering\arraybackslash}p{#1}}

\newtheorem{theorem}{Theorem}
\newtheorem{lemma}{Lemma}
\newtheorem{prop}{Proposition}

\theoremstyle{remark}
\newtheorem{remark}{Remark}

\begin{document}
\title{
Optimizing Task-Specific Timeliness With Edge-Assisted Scheduling for Status Update
}
\author{Jingzhou Sun, Lehan Wang, Zhaojun Nan, Yuxuan Sun, \\Sheng Zhou, Zhisheng Niu~\IEEEmembership{Fellow,~IEEE}
    \thanks{
    This work was sponsored in part by the Natural Science Foundation of China (No. 62341108, No. 62022049, No. 62111530197, No. 62301024, No. 62221001), Hitachi Ltd, the Beijing Natural Science Foundation under grant L222044, the Fundamental Research Funds for the Central Universities under grant 2022JBXT001, and the Talent Fund of Beijing Jiaotong University under grant 2023XKRC030. (\emph{Corresponding author: Sheng Zhou}).


    J. Sun, L. Wang, Z. Nan S. Zhou, and Z. Niu are with Beijing National Research Center for Information Science and Technology, Tsinghua University, Beijing 100084, China. (Emails: sunjz14@tsinghua.org.cn, \{wang-lh19@mails.,nzj660624@mail.,sheng.zhou@, niuzhs@\}tsinghua.edu.cn. )

    Y. Sun is with the School of Electronic and Information Engineering, Beijing Jiaotong University, Beijing 100044, China. (Email: yxsun@bjtu.edu.cn)
 }
}
    \maketitle

\begin{abstract}
Intelligent real-time applications, such as video surveillance, demand intensive computation to extract status information from raw sensing data. This poses a substantial challenge in orchestrating computation and communication resources to provide fresh status information. 
In this paper, we consider a scenario where multiple energy-constrained devices served by an edge server. To extract status information, each device can either do the computation locally or offload it to the edge server. A scheduling policy is needed to determine when and where to compute for each device, taking into account communication and computation capabilities, as well as task-specific timeliness requirements. To that end, we first model the timeliness requirements as general penalty functions of Age of Information (AoI).
A convex optimization problem is formulated to provide a lower bound of the minimum AoI penalty given system parameters.
Using KKT conditions, we proposed a novel scheduling policy which evaluates status update priorities based on communication and computation delays and task-specific timeliness requirements.
The proposed policy is applied to an object tracking application and carried out on a large video dataset. Simulation results show that our policy improves tracking accuracy compared with scheduling policies based on video content information.
\end{abstract}


\section{Introduction}

Fueled by recent advances in wireless communication and computation technologies, cyber-physical network applications have evolved to intelligently connect the physical and cyber worlds, enabled by fully utilizing computation resources scattered over communication networks. These applications, including autonomous driving, remote healthcare, and real-time monitoring, rely on collecting raw sensing data about the time-varying physical environment, extracting valuable status information through computation, and generating control demands based on the status information. However, since the environment changes constantly, control quality degrades until a new status update is made. Hence, the performance of these applications heavily depends on the freshness of status information provided by the network system, which necessitates a shift of focus from solely conveying information bits to providing timely information for certain tasks under service, named as task-oriented communications\cite{gunduz2022beyond}.

One key challenge in this shift is how to effectively orchestrate communication and computation resources in the system, taking account of task-specific timeliness requirements. Over the past decade, edge computing has received much attention due to its potential in providing timely information processing service\cite{chen2018computation}. 
This trend motivates design of schemes that can adaptively offload the computation burden to the edge side or execute it locally, according to the capabilities of both communication and computation resources\cite{mao2016dynamic,dinh2017offloading,kao2017hermes}.

In this work, we study a system consisting of multiple energy-constrained devices, where each device observes a time-varying process and generates tasks that require computation to extract status information about the underlying process. These computation-intensive tasks can be executed on-device or offloaded to an edge server for assistance. Based on the latest status information, devices can decide control actions. And the control quality depends on the freshness of status information. For example, in Augmented Reality (AR) applications, a headset needs to continuously capture the position of certain objects and render virtual elements\cite{chen2019deep}. If the position information becomes stale, the virtual overlay may not align with the physical objects. In this case, the headset can update position status by running object detection algorithms, which can be offloaded to an edge server when the network is in good condition, or done locally at the cost of higher latency and energy consumption. 

Local computing usually takes longer time than that on the edge server side \cite{wang2020joint}, and thus the problem of where to compute seems trivial for single-device system. However, in a multi-device scenario where devices share the wireless channel, some devices might be more suitable than others to offload, probably due to factors such as better channel conditions. Because of limited communication resources, a \emph{scheduling policy} is needed to determine \emph{when devices should generate computation tasks and where computation should be executed to extract status information}.

\subsection{Related Work}
In recent years, \emph{Age of Information} (AoI) has provoked great interest as a metric to quantify the freshness of information \cite{kaul2012real}. In a status update system, AoI measures the elapsed time since the generation of the freshest received information and characterizes the freshness of the status information used for decision-making. Unlike metrics such as delay and throughput, which focus on packet-level performance, AoI provides a system-level view. 

There has been a growing body of research on designing device scheduling policies based on AoI. Among them, weighted average AoI is widely adopted as the optimization target. Periodic status sampling is investigated in \cite{kadota2018scheduling}, and stochastic sampling is considered in \cite{sun2019closed}, where Whittle's index has been shown to enjoy a close-to-optimal performance. Energy harvesting system is studied in \cite{bacinoglu2017scheduling,arafa2017age,arafa2019age}. 

Besides weighted average AoI, nonlinear functions of AoI have also gained attention. In networked control systems with estimation error as the control performance, it has been found that the error can be expressed as a nonlinear function of AoI \cite{champati2019performance,mamduhi2020freshness} for linear time-invariant system, if the sampling time is independent of the underlying status. In the single device case with a general monotonic AoI penalty function, the optimal sampling problem is studied in \cite{sun2019sampling,sun2019sampling2}. For the multi-device case, a scheduling policy using Whittle's index is proposed in \cite{tripathi2019whittle}. A threshold-type policy is derived in \cite{jiang2021analyzing} based on the steady distribution of AoI. A more recent work \cite{shisher2022does} shows that AoI functions can be applied to time-series prediction problem. However, most of these studies have ignored the role of computation in providing fresh information.

As for communication and computation co-design for AoI under the framework of edge computing, 
tandem queuing model is widely adopted to describe the interplay between communication and computation \cite{kuang2020analysis,zou2021optimizing,chiariotti2021peak,qin2022timeliness}. With Poisson sampling process, the average AoI is derived as a function of the sampling rate, transmission rate, and computation rate. Soft update is proposed in \cite{bastopcu2019minimizing} to characterize of process of computation. Optimal sampling policies are derived for exponentially and linearly decaying age cases. In \cite{li2021age}, constrained Markov decision process is adopted to decide when to offload computation-intensive status update to the edge server. In \cite{song2019age}, a finite horizon problem is formulated to optimization linear AoI target. Multi-device scheduling problem is studied in \cite{tripathi2021computation} which only considers local computing.

\subsection{Contributions}

For multi-device scheduling in the context of information freshness, an important problem is how to orchestrate communication and computation resources. Most previous papers on this problem only consider single-device case. The one most closely related to our work is \cite{tripathi2021computation}, but we extend the choices of computing to include the edge side. Our work aims to address the problem of scheduling energy-constrained devices with nonlinear AoI penalty functions and explore ways to provide up-to-date status information by switching between edge computing and local computing. Our contributions can be summarized as follows,
\begin{itemize}
	\item We develop a general framework to jointly consider the communication and computation aspects of real-time status update applications. Computation tasks can be done on-device or on the edge server side. Taking transmission and computation time into account, control performance is modeled as general monotonic AoI penalty functions.
	Given system parameter, a nontrivial lower bound of the time average AoI penalty is derived. By inspecting the property of the lower bound, we propose indices that represent the priority of local computing or edge computing at different AoI values.
	\item A low-complexity scheduling policy is proposed by combining the indices introduced above with the virtual queue technique from Lyapunov optimization \cite{neely2010stochastic}. We show that this policy satisfies the energy constraints of each device. For penalty functions of the form $f(x)=x^p$, $p>0$, we derive the performance gap between the proposed policy and the lower bound, when communication and computation stages take single time slot.
	\item Extensive simulations are carried out to evaluate the performance of the proposed policy for different forms of penalty functions and latency distributions. Simulation results demonstrate that the average AoI penalty under the proposed policy is close to the lower bound. Moreover, we apply the proposed policy to object tracking applications which can be naturally cast as status update processes. The proposed policy is examined on a large video dataset ILSVRC17-VID \cite{ILSVRC15}. Our results show that the proposed policy improves object tracking accuracy by 27\% to that of the video content matching-based scheduling. Furthermore, the proposed policy also outperforms content-based scheduling that has access to the ground truth information. 
\end{itemize}

The rest of this paper is organized as follows. In Section \ref{sec-1}, we present the system model and the problem formulation. In Section \ref{sec-2}, we formulate a convex optimization problem to compute a nontrivial lower bound of the average AoI penalty. In Section \ref{sec-3}, a low-complexity scheduling policy is provided based on the lower bound problem. In Section \ref{sec-4}, numerical results are presented along with the object tracking application. We conclude the paper in Section \ref{sec-5}.

\section{System Model}
\label{sec-1}

\begin{figure}[t]
    \centering
    \includegraphics[width=0.4\textwidth]{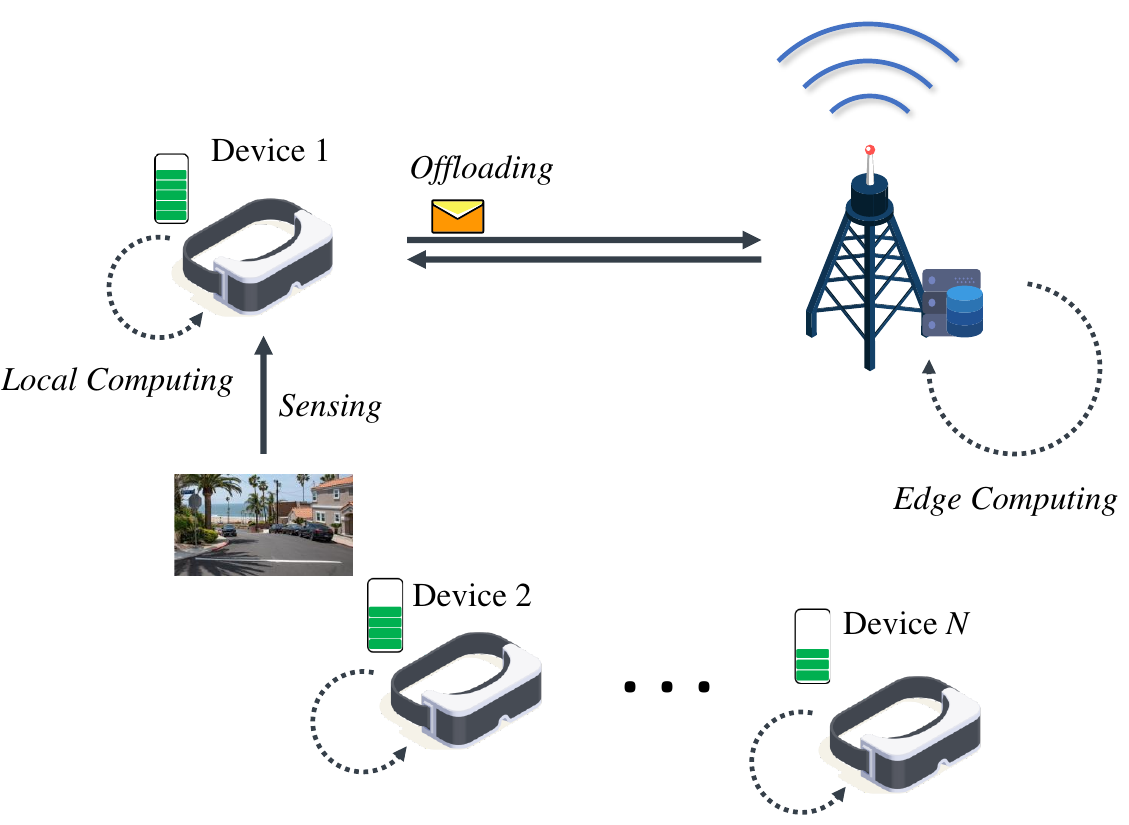}
    \caption{System model.}
    \label{fig:system}
    \vspace{-2em}
\end{figure}

We consider the status update system shown in Fig. \ref{fig:system}. This system consists of a set of energy-constrained devices, denoted as $\mathcal{N}$, with a total number of $N$.  Each device performs a sensing-control task by collecting sensing data, extracting status information from it, and determining control actions based on the status. As the status information becomes stale, the control quality decreases. This simplified model is well-suited for many real-time applications and abstracts away irrelevant details. For energy-constrained mobile devices, however, frequent status updates can quickly drain the battery. Therefore, a scheduling policy is required for each device to decide 1) when to generate a computation task for status update and 2) where to execute the computation.

Slotted time system is considered. At the beginning of a time slot, each device collects sensing data if scheduled, which is assumed to take negligible time. The scheduled devices then perform computation locally or offload computation tasks to an edge server, which takes several slots to finish. For device $n$, $n\in\mathcal{N}$, let $D_{l,n}$ be the number of slots required to finish local computing. The offloading stage takes $D_{t,n}$ slots to send the raw sensing data to the edge server, followed by edge computing that lasts $D_{e,n}$ slots. Result feedback delay is ignored. These three are random variables with finite expectations denoted as $\overline{D}_{l,n}$, $\overline{D}_{t,n}$, and $\overline{D}_{e,n}$, respectively. Furthermore, the latency in each communication or computation stage is assumed to be independent.

We use three binary indicators $u_{l,n}(k)$, $u_{t,n}(k)$, and $u_{e,n}(k)$ to indicate the stage device $n$ is in at time slot $k$. For example, $u_{l,n}(k)=1$ if device $n$ is performing local computing at time slot $k$. Otherwise, $u_{l,n}(k)=0$. Similarly, $u_{t,n}(k)$ and $u_{e,n}(k)$ are associated with the offloading and edge computing stages, respectively. Consider non-preemptive policies, we require that $u_{l,n}(k)+u_{t,n}(k)+u_{e,n}(k)\le 1$. When the summation is zero, device $n$ is idle.

Let $M$ be the number of orthogonal sub-channels. If device $n$ is scheduled to offload sensing data to the edge server, it will occupy one idle channel for $D_{t,n}$ consecutive slots to complete the transmission. On the edge server side, we assume that it is equipped with multi-core hardware and can process multiple computation tasks in parallel \cite{mao2016dynamic}. Therefore, each offloaded computation task is served immediately upon arrival, and there is no queuing delay. 

Let $d_n(k)$ be the latency since the time slot when the sensing data is collectd. 
If $u_{l,n}(k)+u_{t,n}(k)+u_{e,n}(k)=1$, which means that device $n$ is performing status update, $d_n(k)=d_n(k-1)+1$. Otherwise, $d_n(k)=0$. 

When computation is finished, a new control action is generated and returned to the device. Generally, the quality of the control action depends on the freshness of the sensing data used to compute it. To capture this freshness, AoI is defined as the time elapsed since the generation time of the sensing data used to compute the current control action. The AoI of device $n$ at time slot $k$ is denoted as $h_n(k)$. As shown in Fig. \ref{fig:AoI}, AoI evolves as:
 \begin{equation}
    h_n(k) = \left\{
        \begin{array}{c l}
            d_n(k-1) + 1, ~&\text{if the computation is } \\
            &\text{finished at slot $k-1$,}\\
            h_n(k-1) + 1, ~&\text{otherwise.}
        \end{array}
    \right.
 \end{equation}

It is pointed out in \cite{champati2019performance} that, for LTI system, the control quality can be cast as a function of AoI if the sampling process is independent of the content of the underlying process. Following this finding, we model the relationship between control quality and AoI as a penalty function $f_n(\cdot)$, representing the degradation in performance due to information staleness. It is required that the penalty increases with AoI. Furthermore, to avoid ill cases, we also require that the expected penalty with latency $D_{l,n}, D_{t,n}, D_{e,n}$ is finite.

\begin{figure}[t]
    \centering
    \includegraphics[width=0.4\textwidth]{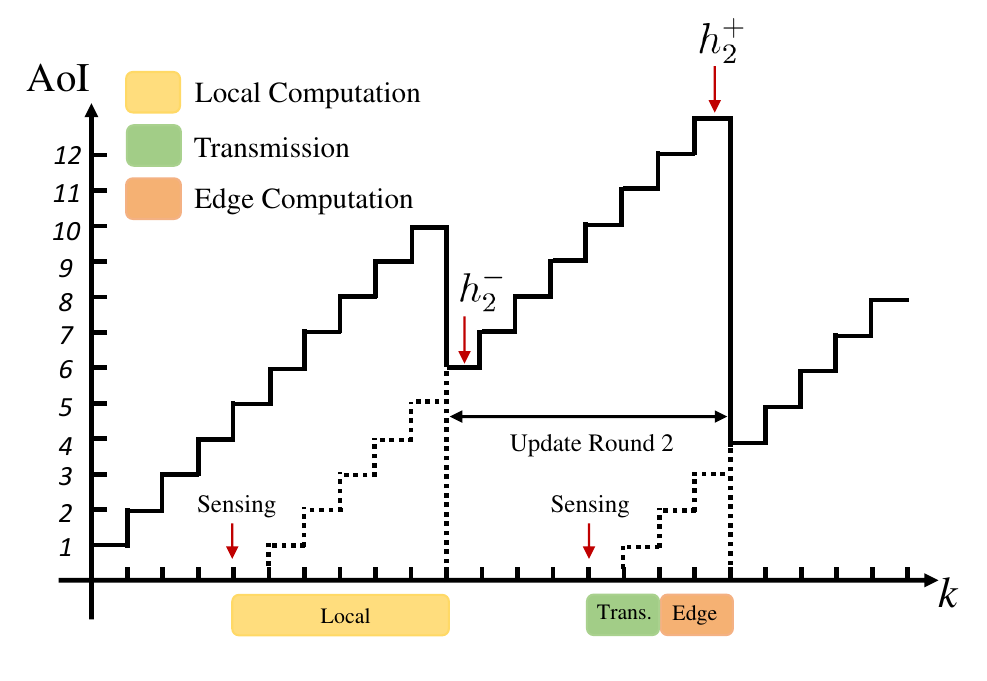}
    \caption{Evolution of AoI.}
    \label{fig:AoI}
    \vspace{-3em}
\end{figure}

We focus on energy consumption on the device side. For local computing, device $n$ takes $E_{l,n}$ Joule per slot. When offloading sensing data to the edge server, device $n$ consumes $E_{t,n}$ Joule per slot. Let $E_n(k)$ be the energy consumption at time slot $k$, it consists of two components $E_n(k) = E_{l,n}u_{l,n}(k) + E_{t,n}u_{t,n}(k)$.
The average energy consumed per time slot by device $n$ should be no larger than $\overline{E}_n$.

Let vector $\bm{h}(k)\triangleq (h_1(k),h_2(k),\dots,h_N(k))$ represent the AoI of all devices at time slot $k$. Similarly, $\bm{d}(k),\bm{u}_{l}(k),\bm{u}_{t}(k),\bm{u}_{e}(k)$ are vectors of corresponding variables. The state of the whole status update system is $\Theta(k)\triangleq (\bm{h}(k),\bm{d}(k),\bm{u}_{l}(k),\bm{u}_{t}(k),\bm{u}_{e}(k))$. The history up to time slot $k$ is denoted as $\mathcal{H}(k)\triangleq \{\Theta(i)|i\le k\}$. A scheduling policy $\pi$ takes in the history $\mathcal{H}(k)$ and decides the new value of $\bm{u}_{l}(k)$ and $\bm{u}_{t}(k)$. 
Note that policy $\pi$ is a centralized policy because it needs $\bm{h}(k)$ and $\bm{d}(k)$ to make decision. To obtain this information, we assume each device will report at the start and end of its computation. Because each device is not always doing computation and this action information is tiny compared to raw sensing data, we ignore this extra cost to implement policy $\pi$. 
Our objective is to propose a scheduling policy that minimizes the time-averaged AoI penalty, subject to energy consumption constraints and communication constraints, as expressed in \textbf{P1}.
 \begin{equation}
 \label{target}
    \begin{aligned}
        \textbf{P1}:\quad \min_{\pi\in \Pi}\quad&
            \sum_{n\in\mathcal{N}}
            \limsup_{K\to\infty}\frac{1}{K}\mathbb{E}_{\pi}
            \left[\sum_{k=1}^{K}f_n(h_n(k))\right]\\
        \textrm{s.t.}\quad 
            & u_{t,n}(k),u_{e,n}(k),u_{l,n}(k) \in \{0,1\}, ~\forall k \ge 1,~\forall n\in\mathcal{N},\\
            & u_{t,n}(k)+u_{e,n}(k)+u_{l,n}(k) \le 1, ~\forall k \ge 1,~\forall n\in\mathcal{N},\\ 
            & \sum_{n\in\mathcal{N}}u_{t,n}(k)\le M, ~\forall k \ge 1,\\
            & \limsup_{K\to\infty} \frac{1}{K}\mathbb{E}_\pi 
                \left[
                    \sum_{k=1}^{K} E_n(k) 
                \right] \le \overline{E}_n, \forall n\in\mathcal{N}.
    \end{aligned}
 \end{equation}
Here, $\Pi$ is the set of non-preemptive policies. This problem can be formulated as a Constrained Markov Decision Process (CMDP) with $\Theta(k)$ representing the state of the system. However, solving this problem exactly is computationally prohibitive. The first reason is that the state space grows exponentially with the number of devices. The second reason is that there are multiple constraints, which renders the standard iteratively tightening approach for CMDP invalid \cite{altman1999constrained}.  Therefore, in the following, we begin by investigating the lower bound of the AoI penalty. Building on this, we propose a low-complexity scheduling policy that draws inspiration from the lower bound problem.

\section{Lower Bound of The AoI Penalty}
\label{sec-2}

In this section, we aim to derive a nontrivial lower bound on the AoI penalty given system parameters. This not only aids in evaluating policy performance but also provides valuable insights into how to design a scheduling policy.

\subsection{Lower Bound Derivation}
We first study the AoI penalty of a single device and then extend the result to multiple devices. For simplicity, the subscript $n$ is dropped temporarily. The time horizon can be divided into disjoint time intervals delineated by the event of computation completion, with each interval being referred to as an \emph{update round}, as shown in Fig. \ref{fig:AoI}. 

Let $h_l^{+}$ and $h_t^{+}$ be the peak age in local computing round and edge computing round respectively. Both are random variables depending on the policy $\pi$. Furthermore, we introduce $\rho_{l}(\pi)$ and $\rho_{t}(\pi)$ to denote the portions of energy spent on local computing and offloading under policy $\pi$ respectively. The following lemma presents an alternative expression for the average AoI penalty in \textbf{P1},
\begin{lemma}\label{thm_0}
    Given policy $\pi\in \Pi$, the average AoI penalty is,
    \begin{equation}
        \begin{aligned}
            &\limsup_{K\to\infty}\frac{1}{K}\mathbb{E}_{\pi}
                \left[\sum_{k=1}^{K}f(h(k))\right]\\
                &= \frac{\rho_{t}(\pi)\overline{E}}{E_{t}\overline{D}_{t}}\mathbb{E}_{\pi}[F(h^{+}_{t})-F(D_{t}+D_{e}-1)] \\
                &+\frac{\rho_{l}(\pi)\overline{E}}{E_{l}\overline{D}_{l}}\mathbb{E}_{\pi}[F(h^{+}_{l})-F(D_{l}-1)],
        \end{aligned}
    \end{equation}
    where
    \begin{align}
        F(h) \triangleq \sum_{x=0}^{h} f(x).
    \end{align} 
\end{lemma}
\begin{proof}
    See Appendix \ref{app-0}.
\end{proof}

Considering the following optimization problem \textbf{P2},
\begin{equation}
\begin{aligned}
        \textbf{P2}: \quad\min_{\pi}\quad  
            &\frac{\rho_{t}\overline{E}}{E_{t}\overline{D}_{t}}\mathbb{E}_{\pi}(\pi)[F(h^{+}_{t})-F(D_{t}+D_{e}-1)]\\
            &+\frac{\rho_{l}(\pi)\overline{E}}{E_{l}\overline{D}_{l}}\mathbb{E}_{\pi}[F(h^{+}_{l})-F(D_{l}-1)] \\
        \text{s.t.}\quad 
            &\frac{\rho_{t}(\pi)\overline{E}}{E_{t}\overline{D}_{t}}(\mathbb{E}_{\pi}[h^{+}_{t}]-(\overline{D}_{t}+\overline{D}_{e}-1)) \\
            &+\frac{\rho_{l}(\pi)\overline{E}}{E_{l}\overline{D}_{l}}(\mathbb{E}_{\pi}[h^{+}_{l}]-(\overline{D}_{l}-1)) = 1.
\end{aligned}
\end{equation}
Lemma \ref{thm_0_1} shows that it provides a lower bound for \textbf{P1},
\begin{lemma}\label{thm_0_1}
    The minimum value of \textbf{P2} is no larger than that of \textbf{P1}.
\end{lemma}
\begin{proof}
    See Appendix \ref{app-0}.
\end{proof}

Because $f(x)$ only takes values at discrete point, we introduce extended penalty function $\tilde{f}(x)$ to facilitate analysis. $\tilde{f}(x)$ is obtained by interpolating $f(x)$ such that: 1) $\tilde{f}(x)$ is an increasing function, 2) $\tilde{f}(x)=f(x)$, when $x\in\mathbb{N}$. As a result, we have
\begin{align}
    \sum_{i=0}^{h} f(i) \ge \int_{0}^{h}\tilde{f}(x)dx.
\end{align}
Let $\tilde{F}(h)$ be the integral of $\tilde{f}(x)$ over $[0,h]$. Then, $\tilde{F}(h)\le F(h)$. Because $\tilde{f}(x)$ is increasing, $\tilde{F}(h)$ is convex. By Jensen's inequality, we have
\begin{align}\label{jensen}
    \mathbb{E}_{\pi}[F(h_l^{+})] \ge \tilde{F}(\mathbb{E}_{\pi}[h_l^{+}]), ~ 
    \mathbb{E}_{\pi}[F(h_t^{+})] \ge \tilde{F}(\mathbb{E}_{\pi}[h_t^{+}]).
\end{align}
Let $x\triangleq \mathbb{E}_{\pi}[h_l^{+}]$ and $ y\triangleq \mathbb{E}_{\pi}[h_t^{+}]$. Plugging \eqref{jensen} into $\textbf{P2}$ leads to the following optimization problem \textbf{P3},
\begin{equation}
\begin{aligned}
        \textbf{P3}:\quad \min_{x\ge 0,y\ge0}\quad & 
            \frac{\rho_{t}(\pi)\overline{E}}{E_{t}\overline{D}_{t}}\left(\tilde{F}(y)-\mathbb{E}_{\pi}[F(D_{t}+D_{e}-1))]\right)\\
            &+\frac{\rho_{l}(\pi)\overline{E}}{E_{l}\overline{D}_{l}}\left(\tilde{F}(x)-\mathbb{E}_{\pi}[F(D_{l}-1)]\right) \\
        \text{s.t.}\quad &
            \frac{\rho_{t}(\pi)\overline{E}}{E_{t}\overline{D}_{t}}(y-(\overline{D}_{t}+\overline{D}_{e}-1))\\
            &+\frac{\rho_{l}(\pi)\overline{E}}{E_{l}\overline{D}_{l}}(x-(\overline{D}_{l}-1)) = 1.
\end{aligned}
\end{equation}
\textbf{P3} relaxes the feasible region of $\mathbb{E}_{\pi}[h_l^{+}]$ and $\mathbb{E}_{\pi}[h_l^{+}]$ to non-negative number, and thus it is a relaxation of \textbf{P2}. The optimal solution is 
\begin{align}\label{upper-opt-single}
    x_{\textrm{opt}} = y_{\textrm{opt}} = 
    \frac{1+ \frac{ \rho_{l}(\pi)\overline{E}(\overline{D}_{l}-1) }{ E_{l}\overline{D}_{l} } + \frac{ \rho_{t}(\pi)\overline{E}(\overline{D}_{t}+\overline{D}_{e}-1) }{ E_{t}\overline{D}_{t} }}{\frac{\rho_{l}(\pi)\overline{E}}{E_{l}\overline{D}_{l}}+\frac{\rho_{t}(\pi)\overline{E}}{E_{t}\overline{D}_{t}}}.
\end{align}

For simplicity, the optimal solution \eqref{upper-opt-single} is denoted as $G(\rho_l,\rho_t)$. Then, with energy proportions $\rho_l,\rho_t$, the penalty lower bound in single-device case is
\begin{equation}
    \begin{aligned}
        &\frac{\rho_{t}\overline{E}}{E_{t}\overline{D}_{t}}\left(\tilde{F}(G(\rho_l,\rho_t))-\mathbb{E}_{\pi}[F(D_{t}+D_{e}-1))]\right) \\
        &+\frac{\rho_{l}\overline{E}}{E_{l}\overline{D}_{l}}\left(\tilde{F}(G(\rho_l,\rho_t))-\mathbb{E}_{\pi}[F(D_{l}-1)]\right).
    \end{aligned}
\end{equation}

Now we zoom out to consider the entire system, and let $\bm{\rho}_{l}\triangleq (\rho_{l,1},\rho_{l,2},\dots,\rho_{l,N})$, $\bm{\rho}_{t}\triangleq (\rho_{t,1},\rho_{t,2},\dots,\rho_{t,N})$. Consider optimization problem \textbf{P4}
\begin{equation}\label{upper-opt}
    \begin{aligned}
        \min_{\bm{\rho}_l,\bm{\rho}_t}\quad&
            \sum_{n\in\mathcal{N}}\left(
            \frac{\rho_{l,n}\overline{E}_n}{E_{l,n}\overline{D}_{l,n}} +
            \frac{\rho_{t,n}\overline{E}_n}{E_{t,n}\overline{D}_{t,n}}
            \right) \tilde{F}_n(G_n(\rho_{l,n},\rho_{t,n})) \\ 
            & -\sum_{n\in\mathcal{N}}\left(\frac{\rho_{t,n}\overline{E}_n}{E_{t,n}\overline{D}_{t,n}} \mathbb{E}[F_n(D_{t,n}+D_{e,n}-1)]\right.\\
            &+ \left.\frac{\rho_{l,n}\overline{E}_n}{E_{l,n}\overline{D}_{l,n}} \mathbb{E}[F_n(D_{l,n}-1)]
            \right)\\
        \textrm{s.t.}\quad 
            &  \sum_{n\in\mathcal{N}} \frac{\rho_{t,n}\overline{E}_n}{E_{t,n}} \le M, \\
            & \rho_{l,n} + \rho_{t,n}\le 1,~\forall n\in\mathcal{N},\\
            & \rho_{l,n}\ge 0, \rho_{t,n}\ge 0, ~\forall n\in\mathcal{N}.
    \end{aligned}
\end{equation}
The first constraint is obtained by relaxing the communication constraint, which originally states that at most $M$ devices can offload simultaneously. It is now relaxed as the time-average number of transmissions, which should not exceed $M$. Therefore, the optimal value of \textbf{P4} provides a lower bound of the time average AoI penalty.

\subsection{Lower Bound Analysis}
In this subsection, we first show that \textbf{P4} is a convex optimization problem, and then study properties of the optimal solution based on KKT conditions.
\begin{lemma}\label{thm_1}
    The optimization problem \textbf{P4} is convex.
\end{lemma}
\begin{proof}
    See Appendix \ref{app-1}.
\end{proof}

For simplicity, let's introduce the following auxiliary variables
\begin{equation}\label{notation-1}
    \begin{aligned}
    &a_n = \frac{\overline{E}_n}{E_{l,n}\overline{D}_{l,n}}, ~
     b_n = \frac{\overline{E}_n}{E_{t,n}\overline{D}_{t,n}}, \\
    &c_n = \overline{D}_{l,n}-1, ~
     d_n = \overline{D}_{t,n}+\overline{D}_{e,n}-1, \\
    &v_n = \mathbb{E}_{\pi}[F_n(D_{l,n}-1)], ~
     w_n = \mathbb{E}_{\pi}[F_n(D_{t,n}+D_{e,n}-1)],\\
    &x_n = \rho_{l,n},~
     y_n = \rho_{t,n}.
\end{aligned}
\end{equation}
Let $\alpha, \bm{\beta},\bm{\gamma},\bm{\nu}$ be Lagrange multipliers. the Lagrangian function is 
\begin{equation}\label{lagrange}
    \begin{aligned}
    L(\bm{x},\bm{y},\alpha,\bm{\beta},\bm{\gamma},\bm{\nu})
    =  & \bm{\beta}^{T}(\bm{x}+\bm{y}-\bm{1})
            - \bm{\gamma}^{T}\bm{x} - \bm{\nu}^{T}\bm{y} \\
            &+
            \sum_{n\in\mathcal{N}}
            (a_n x_n + b_n y_n) \tilde{F}_n(G_n(x_n,y_n))\\
            &- 
            \sum_{n\in\mathcal{N}}(a_n v_n x_n + b_n w_n y_n) \\
            & + \alpha \left( \sum_{n\in\mathcal{N}} \frac{y_n\overline{E}_n}{E_{t,n}}-M \right).
    \end{aligned}
\end{equation}

Because \textbf{P4} is convex, the optimal solution and Lagrange multipliers satisfy KKT conditions. Let $\bm{x}^*, \bm{y}^*, \alpha^*, \bm{\beta}^*,\bm{\gamma}^*,\bm{\nu}^*$ be the corresponding optimal solution. Applying KKT conditions to \eqref{lagrange} provides the following property,
\begin{theorem}\label{thm-lagrange}
    The optimal solution specified by KKT conditions satisfies
    \begin{equation}\label{KKT-final-2}
    \frac{W_{t,n}(h_{t,n}) }{E_{t,n}\overline{D}_{t,n}} -  \frac{W_{l,n}(h_{l,n})}{E_{l,n}\overline{D}_{l,n}}  - \frac{\gamma_n^* - \nu_n^*}{\overline{E}_n} = \frac{\alpha^*}{E_{t,n}},
    \end{equation}
    where
    \begin{align}
    &h_{t,n} = G_n(x_n^*,y_n^*) - (\overline{D}_{t,n}+\overline{D}_{e,n}-1),\\
    &h_{l,n} = G_n(x_n^*,y_n^*) - (\overline{D}_{l,n}-1),
    \end{align}
    and 
    \begin{align}
        &W_{t,n}(x) \nonumber\\
        &= x\tilde{f}_n(x+\overline{D}_{t,n}+\overline{D}_{e,n}-1) \nonumber \\
        &- (\tilde{F}_n(x+\overline{D}_{t,n}+\overline{D}_{e,n}-1)-\mathbb{E}[F_n(D_{t,n}+D_{e,n}-1)]), \label{eq1-3}\\
        &W_{l,n}(x) \nonumber\\
        &= x\tilde{f}_n(x+\overline{D}_{l,n}-1) \nonumber\\
        &- (\tilde{F}_n(x+\overline{D}_{l,n}-1)-\mathbb{E}[F_n(D_{l,n}-1)]). \label{eq1-4}
    \end{align}
\end{theorem}
\begin{proof}
    See Appendix \ref{app-1-2}.
\end{proof}
Here, $h_{t,n}$ represents the expected AoI when device $n$ is scheduled to offload, and $h_{l,n}$ represents the expected AoI when device $n$ is scheduled to do local computing. An intuitive illustration of $W_{t,n}$ and $W_{l,n}$ is shown in Fig. \ref{fig:Meaning}. Taking $W_{l,n}$ as an example, $h$ is the AoI when the scheduling decisions are made, and $d$ is the computing latency. When $x=h$, the first term in \eqref{eq1-4} is the summation of regions I, II, and III, and the second term is the summation of regions I and II. Thus, $W_{l,n}$ is the colored region III. With this geometrical interpretation, the influence of latency and penalty function is reduced to the area of the colored region in Fig. \ref{fig:Meaning}.
\begin{figure}[htbp]
    \centering
    \includegraphics[width=0.35\textwidth]{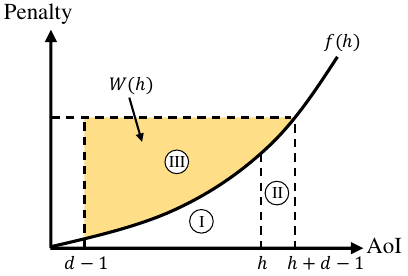}
    \caption{An illustration of $W(x)$.}
    \label{fig:Meaning}
\end{figure}

\begin{remark}
    Rethinking \eqref{KKT-final-2}, the first term $\frac{W_{t,n}(h)}{E_{t,n}\overline{D}_{t,n}}$ is the priority of doing status update by offloading when AoI is $h$. And the second term $\frac{W_{l,n}(h)}{E_{l,n}\overline{D}_{l,n}}$ corresponds to the priority of doing local computing. The third term is related to two Lagrange multipliers: $\gamma_n^*$ and $\nu_n^*$. Due to complementary slackness, $\gamma_n^*=0$ if $\rho^*_{l,n}>0$. For the same reason, $\nu_n^*=0$ if $\rho^*_{t,n}>0$. Note that $\gamma_n^*$ and $\nu_n^*$ can not be larger than $0$ simultaneously.
    When $\gamma_n^*$ and $\nu_n^*$ equal $0$, \eqref{KKT-final-2} is reduced to 
    \begin{align}\label{eq2-1}
        \frac{1}{\overline{D}_{t,n}} W_{t,n}(h_{t,n}) -  \frac{E_{t,n}}{E_{l,n}\overline{D}_{l,n}} W_{l,n}(h_{l,n}) = \alpha^*.
    \end{align}
    Taking $\alpha^*$ as the price of using the channel to offload, Equation \eqref{eq2-1} can be used to determine whether to perform local computing or offload updates.
\end{remark}


\section{Scheduling Policy}
\label{sec-3}

Based on \eqref{KKT-final-2}, which characterizes the expected AoI at scheduling instants, a natural and intuitive policy is to schedule local computing for device $n$ when its AoI is $h_{l,n}$ and to schedule offloading when the AoI is $h_{t,n}$. However, the challenge of obtaining the values of $h_{l,n}$ and $h_{t,n}$, as well as the parameters $\gamma_n$, $\nu_n$, and $\alpha$ at runtime, renders this policy impractical. Nevertheless, the insights provided by \eqref{KKT-final-2} indicate that a scheduling policy should steer the AoI at scheduling instants towards values that align with \eqref{KKT-final-2}. This helps to design scheduling policy for the original problem \textbf{P1}.

We first introduce an auxiliary variable $Q_n$, and rearrange \eqref{KKT-final-2} as 
\begin{equation}
 \label{eq4-1}
    \begin{aligned}
    &\left(\frac{W_{t,n}(h_{t,n})}{\overline{D}_{t,n}} - E_{t,n} Q_n\right) + \frac{E_{t,n}}{E_{l,n}}\left(\frac{W_{l,n}(h_{l,n})}{\overline{D}_{l,n}} - E_{l,n} Q_n\right) \\  
    &= \alpha + \frac{E_{t,n}}{\overline{E}_n}(\gamma_n-\nu_n).
    \end{aligned}
\end{equation}

If $Q_n$ satisfies
\begin{align} \label{eq4-1}
    \frac{W_{l,n}(h_{l,n})}{\overline{D}_{l,n}} - E_{l,n} Q_n = 0,
\end{align}
then
\begin{align}\label{eq4-2}
    \frac{W_{t,n}(h_{t,n})}{\overline{D}_{t,n}} - E_{t,n} Q_n  = \alpha + \frac{E_{t,n}}{\overline{E}_n}(\gamma_n-\nu_n) .
\end{align}
If the value of $Q_n$ is known, \eqref{eq4-1} and \eqref{eq4-2} provide a heuristic scheduling policy.
Firstly, for the edge computing part, since the function $W_{t,n}(x)$ is increasing, we can sort idle devices in descending order based on the left-hand side of \eqref{eq4-2} with $h_{t,n}$ replaced by $h_n(k)$ and select no more than $m(k)$ devices to offload, where $m(k)$ is the number of idle channels at time slot $k$. Then, if device $n$ is still idle and $h_n(k)$ satisfies that 
\begin{align}\label{eq4-1-1}
    \frac{W_{l,n}(h_{n}(k))}{\overline{D}_{l,n}} \ge E_{l,n} Q_n,
\end{align}
it will be scheduled to do local computing.
By adopting this approach, we can bring the AoI at scheduling instants closer to the values specified by \eqref{KKT-final-2}.

$Q_n$ plays the role of threshold in this policy, such that devices will not update so frequently that the energy constraints are violated. In other words, $Q_n$ is determined by energy constraints. Although it is hard to calculate the exact value of $Q_n$, we can approach it at runtime. Based on this insight, we use tools from Lyapunov optimization \cite{neely2010stochastic} and introduce virtual queue $Q_n(k)$:
\begin{align}\label{evq}
    Q_n(k+1) \triangleq \max\{Q_n(k) - \overline{E}_{n} + E_n(k),0\},
\end{align}
$Q_n(k)$ corresponds to the energy consumption until time slot $k$. If update is too frequent, $Q_n(k)$ will increase and prevent further update. As system evolves, $Q_n(k)$ approximates $Q_n$.

Let $\mathcal{N}_{\text{idle}}(k)$ be the set of devices that are idle at the beginning of time slot $k$. The two auxiliary sets are defined as  
\begin{align*}
    &\mathcal{C}_l(k) \triangleq \left\{n~\left|~\frac{W_{l,n}(h_n(k))}{\overline{D}_{l,n}}\ge VE_{l,n}Q_n(k), ~n\in\mathcal{N}_{\text{idle}}(k)\right.\right\}, \\
    &\mathcal{C}_t(k) \triangleq \left\{n~\left|~\frac{W_{t,n}(h_n(k))}{\overline{D}_{t,n}}\ge VE_{t,n}Q_n(k), ~n\in\mathcal{N}_{\text{idle}}(k)\right.\right\},
\end{align*}
where $V$ is a parameter used to smooth the fluctuation of $Q_n(k)$.

The set $\mathcal{C}_l(k)$ consists of devices eligible for local computing, while $\mathcal{C}_t(k)$ consists of devices eligible for edge computing. The intersection of these two sets may not be empty. To simplify the expression of scheduling policy, we introduce index $I_n(k)$ as 
\begin{equation}\label{index}
\begin{aligned}
        &I_n(k) =\\
         &\left\{\begin{array}{l l}
        \frac{W_{l,n}(h_n(k))}{\overline{D}_{l,n}} - V E_{l,n}Q_n(k),~&\text{if $n\in\mathcal{C}_l(k) - \mathcal{C}_t(k)$,}\\
        \frac{W_{t,n}(h_n(k))}{\overline{D}_{t,n}} - V E_{t,n}Q_n(k),~&\text{if $n\in\mathcal{C}_t(k) - \mathcal{C}_l(k)$,}\\
        \frac{W_{t,n}(h_n(k))}{\overline{D}_{t,n}} - \frac{W_{l,n}(h_n(k))}{\overline{D}_{l,n}} ~&\text{if $n\in\mathcal{C}_l(k) \cap \mathcal{C}_t(k)$.}\\
        + V(E_{l,n}-E_{t,n})Q_n(k),~&
    \end{array}\right.
\end{aligned}
\end{equation}

Based on the insights from \eqref{KKT-final-2}, we propose a \emph{Max-Weight scheduling policy} $\pi_{\text{MW}}$, which makes scheduling decisions at each time slot as shown in algorithm \ref{Alg-0}. In this algorithm, the scheduler first decides which devices to offload based on their values of $I_n(k)$, which is derived from \eqref{eq4-2}. Subsequently, those devices that are still idle will perform local computing if they fall within the set $\mathcal{C}_{l}(k)$, as dictated by \eqref{eq4-1}.
\begin{algorithm}[htp]
\caption{Max-Weight scheduling}
\begin{algorithmic}[1]
\REQUIRE
    The number of idle channels $m(k)$.\\
\STATE Sort devices in $\mathcal{C}_t(k)$ in descending order according to the value of $I_n(k)$. The result is $(n_1,n_2,\dots,n_S)$. $S$ is the total number of devices in this set.
\STATE $s \gets 1$
\WHILE {$s\le S$}
    \IF{$m(k)>0$ and $I_{n_s}(k)\ge 0$}
        \STATE $u_{t,n_s}(k)\gets 1$, $m(k)\gets m(k)-1$
    \ELSIF{$n_s \in \mathcal{C}_l(k) \cap \mathcal{C}_t(k)$}
        \STATE $u_{l,n_s}(k)\gets 1$
    \ENDIF
    \STATE $s \gets s+1$
\ENDWHILE
\FOR{$n\in\mathcal{C}_l(k)-\mathcal{C}_t(k)$}
        \STATE $u_{l,n}(k)\gets 1$
\ENDFOR
\end{algorithmic}
\label{Alg-0}
\end{algorithm}

The term $VQ_n(k)$ plays the role of $Q_n$ in \eqref{eq4-1} and \eqref{eq4-2}. Since we want $\limsup_{k\to\infty}VQ_n(k)$ and $\liminf_{k\to\infty}VQ_n(k)$ to be close to $Q_n$, it is expected that a smaller $V$ enjoys a better performance, because a small $V$ can smooth fluctuations in the value of $Q_n(k)$. This conjecture is substantiated in Section V.

Theorem \ref{thm-MW} demonstrates that algorithm \ref{Alg-0} maximizes a term that is linear in $u_{l,n}$ and $u_{t,n}$, as follows:
\begin{theorem}\label{thm-MW}
    Algorithm \ref{Alg-0} makes scheduling decisions $\bm{u}_{l}(k)$ and $\bm{u}_t(k)$ to maximize the following,
    \begin{equation}
        \label{eq-weight}
        \begin{aligned}
            &\sum_{n\in\mathcal{N}_{\normalfont\text{idle}}(k)} \left(\frac{W_{l,n}(h_n(k))}{\overline{D}_{l,n}} - V E_{l,n}Q_n(k)\right) u_{l,n}(k)\\
            &+\sum_{n\in\mathcal{N}_{\normalfont\text{idle}}(k)} \left(\frac{W_{t,n}(h_n(k))}{\overline{D}_{t,n}} - V E_{t,n}Q_n(k)\right) u_{t,n}(k).
        \end{aligned}
    \end{equation}
\end{theorem}
\begin{proof}
 See Appedix \ref{app-2}.
\end{proof}

The following theorem establishes that, under mild assumptions, policy $\pi_{\text{MW}}$ satisfies the energy constraints in \eqref{target}.

\begin{theorem}\label{thm_4}
    For any $n\in\mathcal{N}$, if there exists $D_n^*$ such that $D_{l,n}$, $D_{t,n}$ and $D_{e,n}$ are smaller than $D_n^*$, then
    \begin{align}
        \limsup_{K\to\infty} \frac{1}{K}\mathbb{E}_{\pi_{\text{MW}}} 
                \left[
                    \sum_{k=1}^{K} E_n(k) 
                \right] \le \overline{E}_n, \forall n\in\mathcal{N}.
    \end{align}
\end{theorem}
\begin{proof}
    Let $k_i$ be the time slot at which device $n$ starts its $i$-th round of local computing or offloading. Because the delay is bounded, we have 
    \begin{align}\label{thm-4-4}
        Q(k_{i+1}) \le Q(k_i) + \max(E_{l,n},E_{t,n})D_n^*.
    \end{align}
    We will prove that there exists $L$ such that 
    \begin{align}\label{thm-4-1}
         \limsup_{i\to\infty} Q_n(k_i) \le L.
    \end{align}

    Given $k_i$, there exists $s_i$ such that 
    \begin{equation}\label{thm-4-2}
        \begin{aligned}
             &Q_n(k_i) \ge Q_n(k_1)+(s_i-1)\max(E_{l,n},E_{t,n})D_n^*, \\
             &Q_n(k_i) < Q_n(k_1)+s_i\max(E_{l,n},E_{t,n})D_n^*. 
        \end{aligned}
    \end{equation}

    Let $s^*=\sup \{s_i,i\ge 1\}$. If $s^*$ is finite, then \eqref{thm-4-1} holds trivially. Otherwise, consider an $i^*$ such that 
    \begin{align}\label{eq4-3}
         Q_n(k_{i^*})-t\overline{E}_n \ge \max\left(\frac{W_{l,n}(2D_n^*+t)}{V E_{l,n}}, \frac{W_{t,n}(2D_n^*+t)}{V E_{t,n}}\right), 
    \end{align}
    and 
    \begin{equation}
    \begin{aligned}
       &t\overline{E}_n + \max\left(\frac{W_{l,n}(2D_n^*+t)}{V E_{l,n}}, \frac{W_{t,n}(2D_n^*+t)}{V E_{t,n}}\right) \\
       &\le Q_n(k_1)+(s_{i^*}-1)\max(E_{l,n},E_{t,n})D_n^*,
    \end{aligned}     
    \end{equation}
    where $t\triangleq \left\lceil\frac{\max(E_{l,n},E_{t,n})D_n^*}{\overline{E}_n}\right\rceil$. 
    \eqref{eq4-3} means that device $n$ is idle for $t$ time slots at least, after the completion of the ${i^*}$-th status update. Therefore,
    \begin{align}\label{thm-4-3}
        Q_n(k_{i^*+1}) \le Q_n(k_{i^*})+\max(E_{l,n},E_{t,n})D_n^*-t\overline{E}_n.
    \end{align}
    According to the definition of $t$, \eqref{thm-4-3} yields that $Q_n(k_{i^*+1}) \le Q_n(k_{i^*})$.
    
    If $Q_n(k_{i^*+1})$ falls in the range specified in \eqref{thm-4-2} with $s_{i^*}$, repeating the analysis above gives that $Q_n(k_{i^*+2})\le Q_n(k_{i^*+1})$. If $Q_n(k_{i^*+1}) < Q_n(k_1)+(s_i^*-1)\max(E_{l,n},E_{t,n})D_n^*$, due to \eqref{thm-4-4}, we have 
    \begin{align}
         Q_n(k_{i^*+2}) \le Q_n(k_1)+s_{i^*}\max(E_{l,n},E_{t,n})D_n^*.
    \end{align} 

    Based on induction, we conclude that 
    \begin{align}
         Q_n(k_{j}) \le Q_n(k_1)+s_{i^*}\max(E_{l,n},E_{t,n})D_n^*, \forall j\ge i^*,
    \end{align} 
    and thus \eqref{thm-4-1} holds.

    Recall the definition of $Q_n(k)$ in \eqref{evq}, we have for all $k\ge 1$:
    \begin{align}
        \frac{1}{K}\sum_{k=1}^{K}E_n(k) \le \frac{Q_n(K+1)}{K}-\frac{Q_n(1)}{K} + \overline{E}_n.
    \end{align}
    Taking expectations of the above and letting $K\to\infty$ yields:
    \begin{align}
        \limsup_{K\to\infty} \frac{1}{K}\mathbb{E}_{\pi^{\text{MW}}} 
                \left[
                    \sum_{k=1}^{K} E_n(k) 
                \right] \le \overline{E}_n.
    \end{align}
\end{proof}

One important distinction between our work and other studies that use Max-Weight policy for scheduling, such as \cite{kadota2018scheduling}, is that the set of idle devices in our problem varies with time. Thus, conventional approaches based on strongly stable queue techniques cannot be applied to our problem directly.

Although a general performance guarantee is difficult to establish, the following proposition provides insight into the performance gap for a special case:

\begin{prop}\label{prop-2}
    Let $J^{\pi_{\text{MW}}}$ be the average AoI penalty under policy $\pi_{\text{MW}}$.
    When the penalty function is $f_n(x)=\alpha_n x^p$, $p>0$, and $D_{l,n} = D_{t,n} = 1$, $D_{e,n} = 0$, $\forall n\in\mathcal{N}$, $J^{\pi_{\text{MW}}}$ satisfies
    \begin{align}\label{performance}
        \left(\frac{J^{\pi_{\text{MW}}}}{p+1}\right)^{p+1} \le J^*\left(\frac{B}{p}+J^{\pi_{\text{MW}}}\right)^{p},
    \end{align}
    where $J^*$ is the lower bound from \eqref{upper-opt}, and $B$ is defined as 
    \begin{align}
        B \triangleq \frac{V}{2}\sum_{n\in\mathcal{N}}(\max(E_{l,n},E_{t,n})-\overline{E}_n)^2.
    \end{align}
\end{prop}
\begin{proof}
    See Appendix \ref{app-prop-2}.
\end{proof}
\begin{remark}\theoremstyle{remark}
    When $p=1$, the penalty function is in linear form, and the target becomes the weighted time average age. Let $p=1$ in \eqref{performance}, we obtain the following inequality:
    \begin{align}
        (J^{\pi_{\text{MW}}}-2J^*)^2 \le 4J^*(B+J^*),
    \end{align}
    which yields,
    \begin{align}
        \frac{J^{\pi_{\text{MW}}}}{J^*} \le 2 + 2\sqrt{\frac{B}{J^*}+1}.
    \end{align}
    This suggests that the weighted average age achieved by the Max-Weight policy is bounded within approximately four times of the lower bound.
\end{remark}

\section{Numerical Results}
\label{sec-4}
In this part, we evaluate the performance of the proposed policy under various settings. In addition to extensive simulations on synthetic data, we also apply this policy to a video tracking task and carry out experiments on ILSVRC17-VID dataset.

\begin{figure*}[t]
\centering
     \begin{subfigure}[b]{0.32\textwidth}
         \centering
         \includegraphics[width=\textwidth]{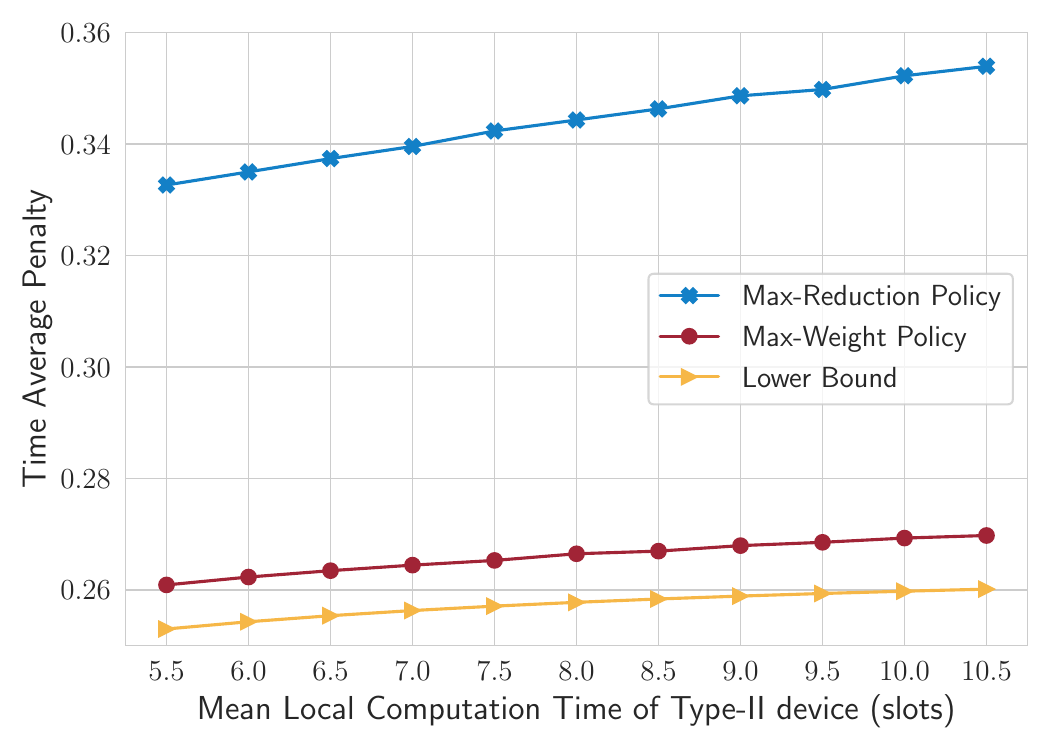}
         \caption{Composite}
         \label{fig:composite}
     \end{subfigure}
     \begin{subfigure}[b]{0.32\textwidth}
         \centering
         \includegraphics[width=\textwidth]{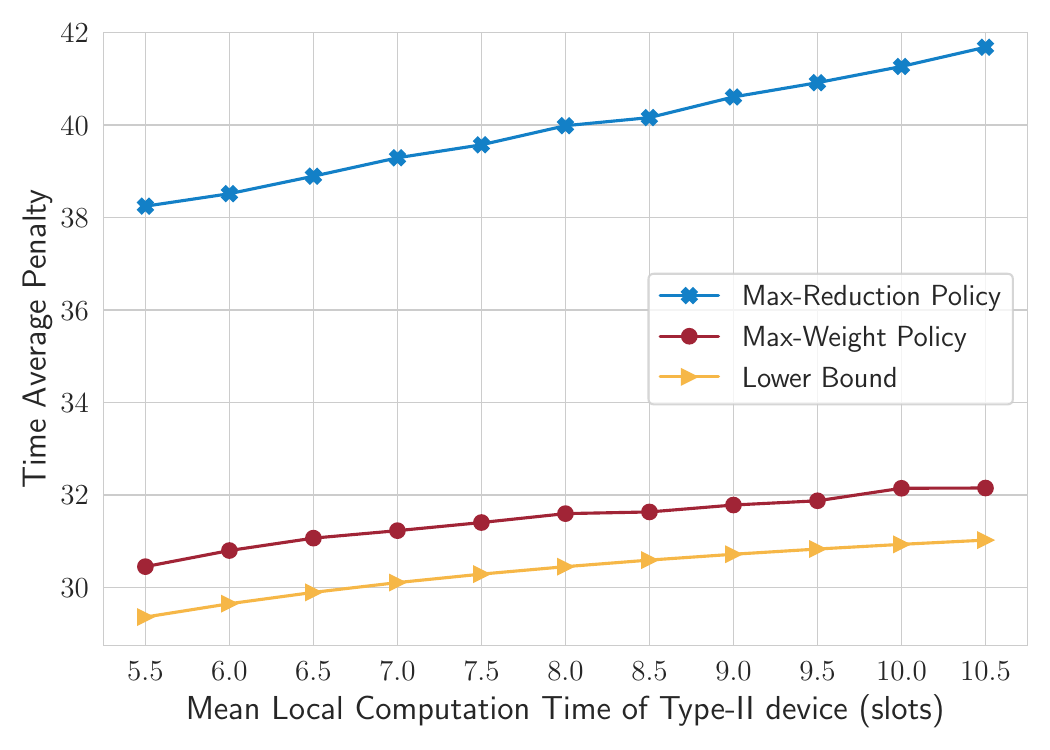}
         \caption{Linear}
         \label{fig:linear}
     \end{subfigure}{}
     \hfill
     \begin{subfigure}[b]{0.32\textwidth}
         \centering
         \includegraphics[width=\textwidth]{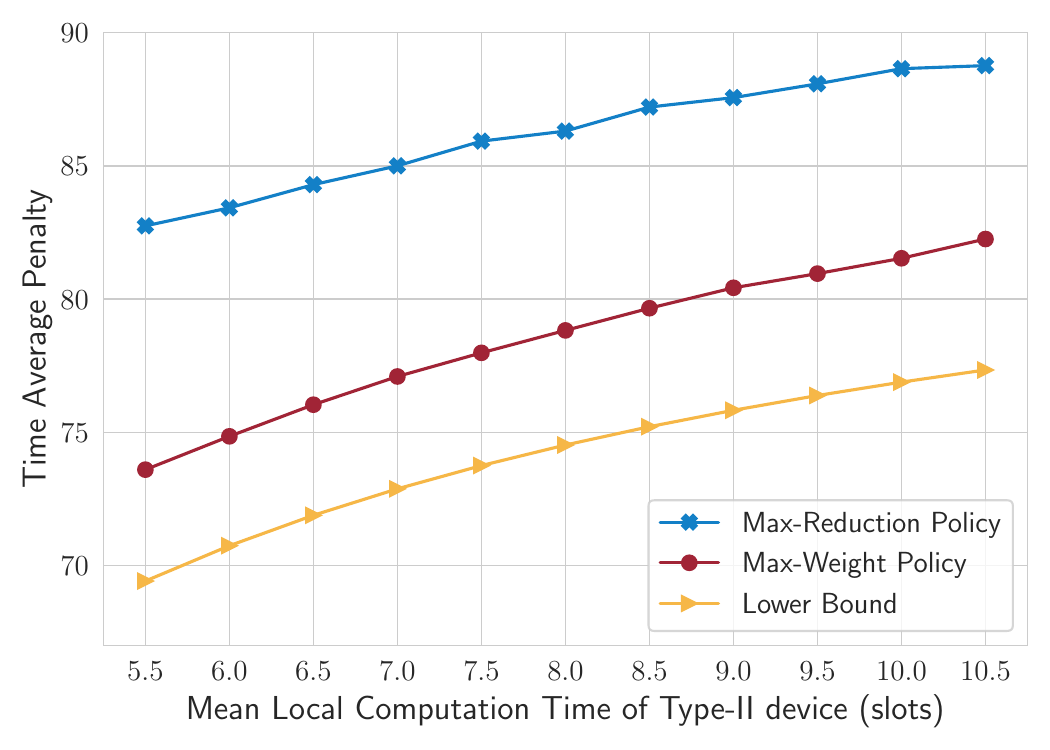}
         \caption{Square}
         \label{fig:square}
     \end{subfigure}
     \hfill
        \caption{Performance comparison.}
        \label{fig:Sim-1}
\end{figure*}

\subsection{Simulation Results}
Scheduling decisions depend on various factors, including the form of penalty function, computation delay, transmission delay, etc. To facilitate experiments, the set of devices is divided into 2 types. Part of the simulation settings are listed in Table \ref{table:setting}, where $U(a,b)$ means taking values uniformly in the set $\{a,a+1,\dots,b\}$. In the first simulation, the delay distribution of Type-II devices' local computing delay follows $U(1,x)$ with $x$ increasing from 10 to 20.

Different kinds of penalty functions are considered in the simulation, including linear function, square function, and a special type of composite function, as shown in Table \ref{table:function}.

By varying the distribution of Type-II devices' local computing delay, we obtain the results shown in Fig. \ref{fig:Sim-1}. The number of devices is $30$. Half of them are Type-I, and the left are Type-II. The number of orthogonal channels is $3$.  
\emph{Max-Reduction policy} \cite{zhong2019age} is considered for comparison. In Max-Reduction policy, the terms $\frac{W_{l,n}(h_n(k))}{\overline{D}_{l,n}}$ and $\frac{W_{t,n}(h_n(k))}{\overline{D}_{t,n}}$ in \eqref{index} are replaced by the expected penalty reduction after scheduling. 
The lower bound is obtained by solving the optimization problem \textbf{P4} numerically. The simulation horizon is $10^6$ slots. $V$ is set to be $0.01$ for composite penalty function, and $1$ for both linear and square penalty functions. 
The performance of the proposed Max-Weight policy is close to the lower bound. It should be noted that the lower bound is derived by using Jenson's inequality, and thus the estimation error between the lower bound and the minimum average AoI penalty gets larger for higher-order penalty functions. 

\begin{table}[!t]
\centering
\caption{Simulation settings.}
\begin{tabular}{@{}c|cc@{}}
\toprule
\textbf{Parameter}                  & \textbf{Type-I}  & \textbf{Type-II} \\ \midrule
Local Comp. Delay (slots)  & $U(1,15)$ & $U(1,x)$       \\
Transmission Delay (slots) & $U(1,3)$  & $U(3,7)$  \\
Edge Comp. Delay (slots)   & $U(1,2)$  & $U(1,2)$  \\
Local Comp. Energy (J/slot)        & 10       & 10       \\
Transmission Energy (J/slot)       & 1       & 1       \\
Energy Budget (J/slot)             & 0.4     & 0.4     \\
\hline
\end{tabular}
\label{table:setting}
\end{table}

\begin{table}[!t]
\centering
\caption{List of penalty functions.}
\begin{tabular}{@{}c|cc@{}}
\toprule
\textbf{Function}                  & \textbf{Type-I}  & \textbf{Type-II} \\ \midrule
Linear  & $x$ & $2x$       \\
Square & $0.1x^2$  & $0.2x^2$  \\
Composite  & $1-(0.02x+1)^{-0.4}$  & $1-(0.14x+1)^{-0.4}$  \\
\hline
\end{tabular}
\label{table:function}
\end{table}

\begin{table}[!t]
\centering
\caption{Coefficient of variation under different penalty functions.}
\begin{tabular}{|c||c|c|c|}
 \hline
 \multicolumn{4}{|c|}{Coefficient of Variation} \\
 \hline
 Penalty Function &  Composite &  Linear & Square\\
 \hline
 Max-Weight   & \textbf{0.129}  &   \textbf{0.062} & \textbf{0.047} \\
 Max-Reduction  & 0.787  &  0.680 & 0.539 \\
 \hline
\end{tabular}
\label{table:KKT}
\vspace{-2em}
\end{table}

It is also interesting to check whether the proposed policy does steer AoI to be aligned with \eqref{KKT-final-2}. Considering the case where the local computing delay of Type-II devices follows $U(1,10)$, we estimate the value of $\alpha$ by plugging the peak AoI value after each computation into \eqref{eq2-1} and calculate the average for each device. And thus we obtain $30$ points, each corresponding to one device. We then calculate the mean value and standard deviation over these $30$ devices. The Coefficient of Variation (CV) is listed in Table \ref{table:KKT}, which is the ratio of the standard deviation to the mean. The result shows that the CV values of Max-Weight policy are one order of magnitude smaller than that of the Max-Reduction policy.

\begin{figure}[t!]
    \centering
    \includegraphics[width=0.45\textwidth]{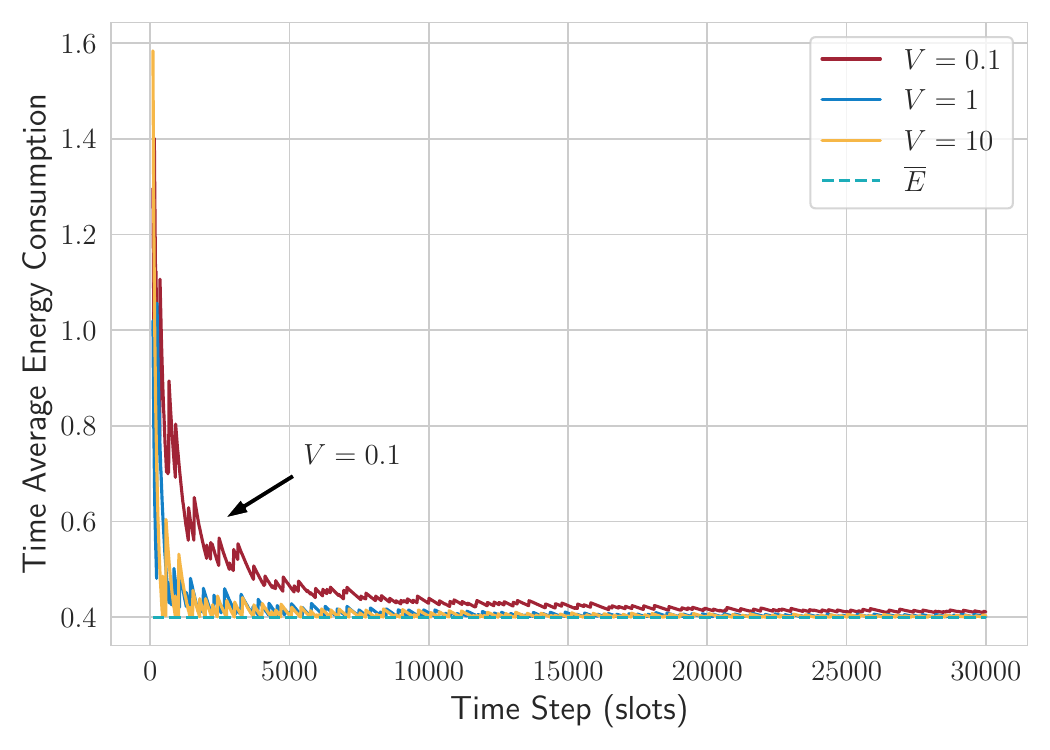}
    \caption{The average energy consumption under different $V$ with square penalty function.}
    \label{fig:energy}
    \vspace{-3em}
\end{figure}

To investigate the influence of parameter $V$, we first check the average energy consumption, as shown in Fig. \ref{fig:energy}. These curves are obtained by running the Max-Weight policy and calculating the moving average of energy consumption. We choose the square penalty function case and plot the first $30000$ time slots. The local computing time for Type-II devices is $U(1,10)$. The cyan dashed line corresponds to the energy budget $0.4$. The first observation is that all three curves converge to the horizontal cyan line, this is in line with Theorem \ref{thm_4}. Another observation is that smaller $V$ results in slower convergence to the expected value. This might be because a larger $V$ means fewer rounds to reach the desired $Q_n$ value.

\begin{figure}[t!]
    \centering
    \includegraphics[width=0.45\textwidth]{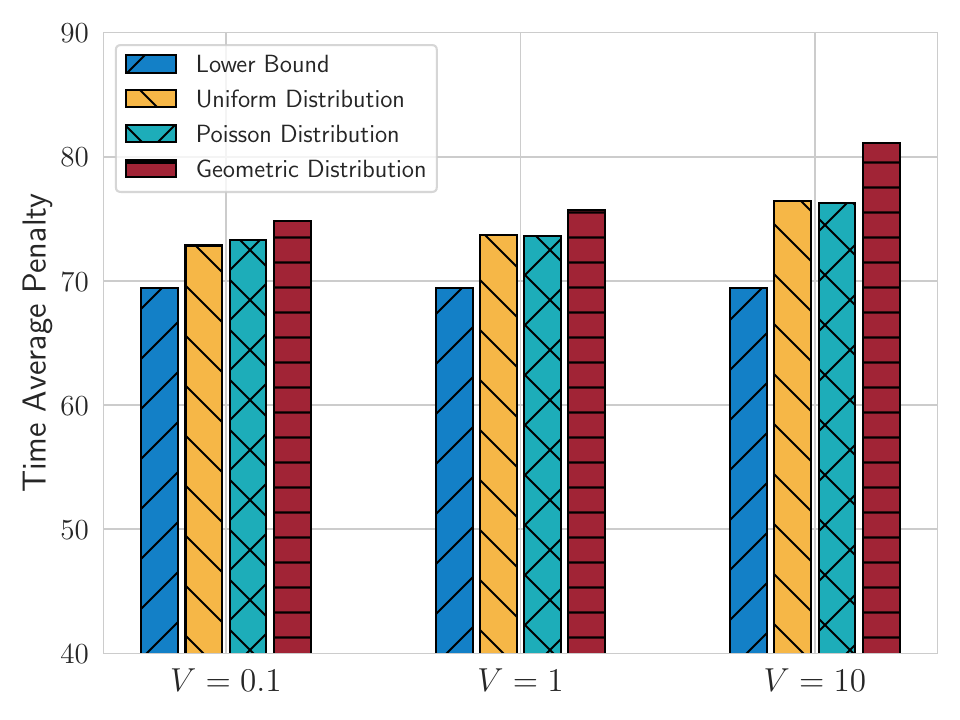}
    \caption{The average penalty under different $V$ and delay distributions.}
    \label{fig:VandDist}
    \vspace{-2em}
\end{figure}

However, the convergence speed comes at the price of performance loss. As shown in Fig. \ref{fig:VandDist}, increasing $V$ from $0.1$ to $10$ leads to larger average penalty. This is because that a larger $V$ increases the fluctuation of the virtual queue $Q(k)$, as discussed in Section \ref{sec-3}. The influence of delay distribution is also studied in Fig. \ref{fig:VandDist}. Fixing the mean value, we run simulations when delay follows uniform distribution, Poisson distribution and geometric distribution respectively. The performance under geometric distribution is the worst. This might be because that the geometric distribution has the largest variance among the three in this case. 

\subsection{Experimental Results}

To show the usage of the proposed policy, we choose an object tracking application for demonstration. Tracking object is key to many visual applications \cite{chen2015glimpse,ran2018deepdecision,liu2019edge,wang2020joint}. Given an object's initial position, the tracker tracks this object as it moves. In this process, tracking error accumulates, and tracking performance would decrease if the tracker has not been refreshed. Fig. \ref{fig:IoU_example} gives an example of the tracking process. The red dash box is the position of the target car, and the blue box is the tracking result. After 30 video frames, the blue box drifts away from the true position. 
\begin{figure}[htbp]
    \centering
    \includegraphics[width=0.45\textwidth]{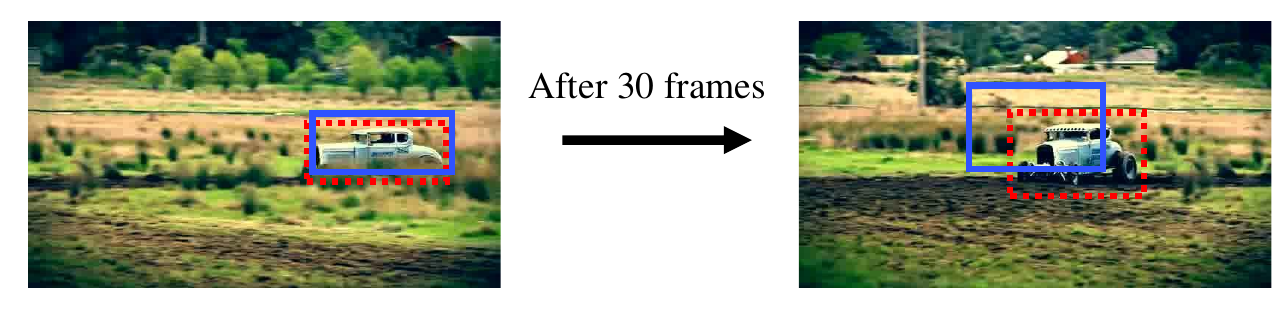}
    \caption{The tracking performance degrades as the object moves.}
    \label{fig:IoU_example}
\end{figure}

To refresh the tracker, object detection algorithms\cite{zhu2017deep,tan2020efficientdet} is called to obtain the current position of the target object. The detection can be done on-device or offloaded to an edge server.
Thus, \emph{object tracking task can be naturally cast as a status update process, where status update refers to the object detection step}. In this case, AoI is defined as the number of video frames since the latest frame used for object detection. 

To evaluate tracking performance, we first calculate the IoU (Intersection over Union). It represents the area of the intersection over that of the union. Let $B_1$ be the tracking position and $B_2$ be the actual position, IoU is defined in \eqref{eq-IoU}. Tracking performance is measured by the probability of the IoU larger than a given threshold $\text{IoU}_\text{th}$: $\mathbb{P}\left(\text{IoU}_{\text{curr}}\ge \text{IoU}_\text{th}\right)$, where $\text{IoU}_{\text{curr}}$ is the IoU of the current frame. 
\begin{align}\label{eq-IoU}
    \text{IoU} = \frac{\text{Area}(B_1\cap B_2)}{\text{Area}(B_1\cup B_2)}.
\end{align}


In this experiment, we choose CSRT algorithm \cite{lukezic2017discriminative} for video tracking, which is faster than DNN-based methods. To isolate influence from the detection algorithm, it is assumed that the detection algorithm can always return the accurate position. 
We first do profiling on ILSVRC17-VID dataset to evaluate the tracking performance as a function of AoI. The IoU thresholds $\text{IoU}_\text{th}$ are set to be 0.5 and 0.75, representing different requirements for tracking accuracy. 

\begin{figure}[htbp]
    \centering
    \includegraphics[width=0.45\textwidth]{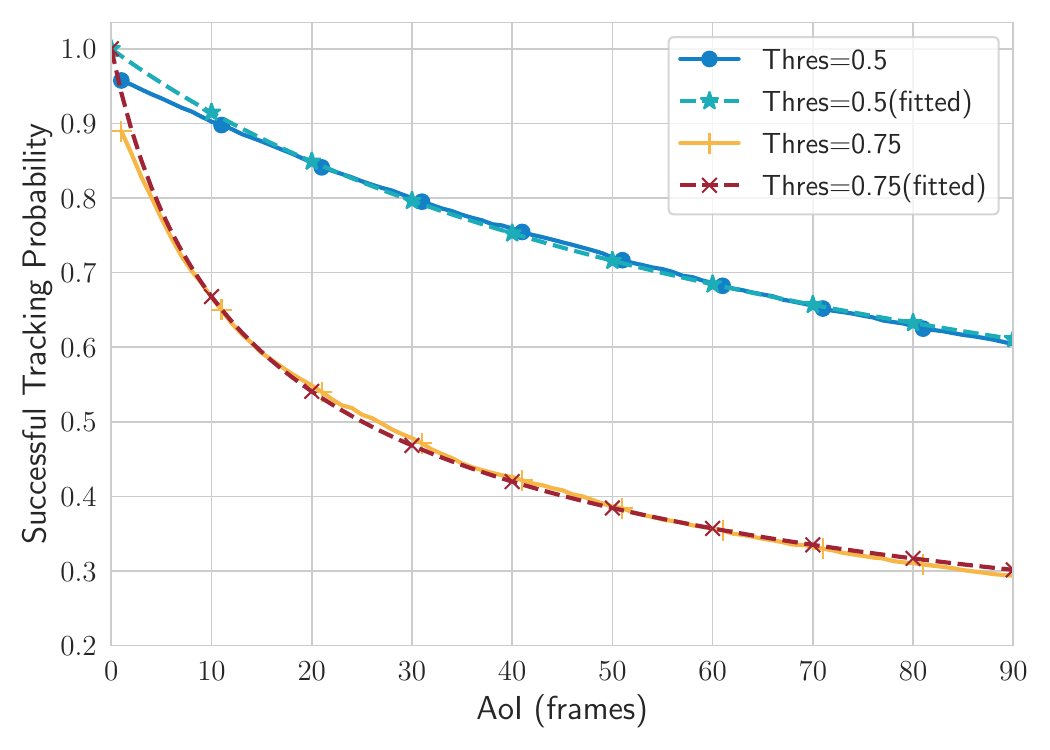}
    \caption{Profiling result of the successful tracking probability as a function of AoI.}
    \label{fig:fitting}
\end{figure}

30\% of the videos in the dataset are chosen for profiling, i.e., 540 videos. For each video, we start from frame 1, initialize the CSRT tracker with bounding boxes, and let it track the following 90 frames. Then, the tracker is refreshed with the actual positions in the 91-st frame and repeats the procedure. Fig. \ref{fig:fitting} shows the profiling result, where these two curves can be fitted by functions of the form of $(ax+1)^{-b}$. Table.\ref{table:fitting} shows the fitted parameters. Thus, the penalty function is modeled as $1-(ax+1)^{-b}$ with the penalty being the tracking failure probability.
\begin{table}[htbp]
\centering
\caption{Fitted parameters with different IoU
thresholds.}
\begin{tabular}{@{}c|cc@{}}
\toprule
\textbf{IoU Requirement}                  & $\bm{a}$  & $\bm{b}$ \\ \midrule
$0.5$  & $0.02149158$ & $0.45788114$       \\
$0.75$ & $0.14155363$  & $0.45766638$  \\
\hline
\end{tabular}
\label{table:fitting}
\end{table}

This experiment is done on a simulator we build on the server. We set the number of tracking devices to be 20, half of which are labeled as Type-I device with $\text{IoU}_\text{th}=0.5$. The other half are labeled as Type-II device with $\text{IoU}_\text{th}=0.75$. As for parameter settings, we set both two types' local computing delay follows Gaussian distribution $\mathcal{N}(200,30)$ms\cite{wang2020joint}, truncated above 0. The local computing power is set to be $2.5$W. For transmission part, the Type-I's transmission delay follows distribution $\mathcal{N}(30,10)$ms, and Type-II's transmission delay follows distribution $\mathcal{N}(60,20)$ms. The transmission power is set to be $250$mW. The energy budget is set to be $300$mW. For the computation delay on the edge side, we test the inference time of Faster-RCNN network\cite{ren2015faster} with ResNet50\cite{he2016deep} as backbone on a Linux server with TITAN Xp GPU. The computation time distribution is shown in Fig. \ref{fig:comp}.
\begin{figure}[!htb]
    \centering
    \includegraphics[width=0.45\textwidth]{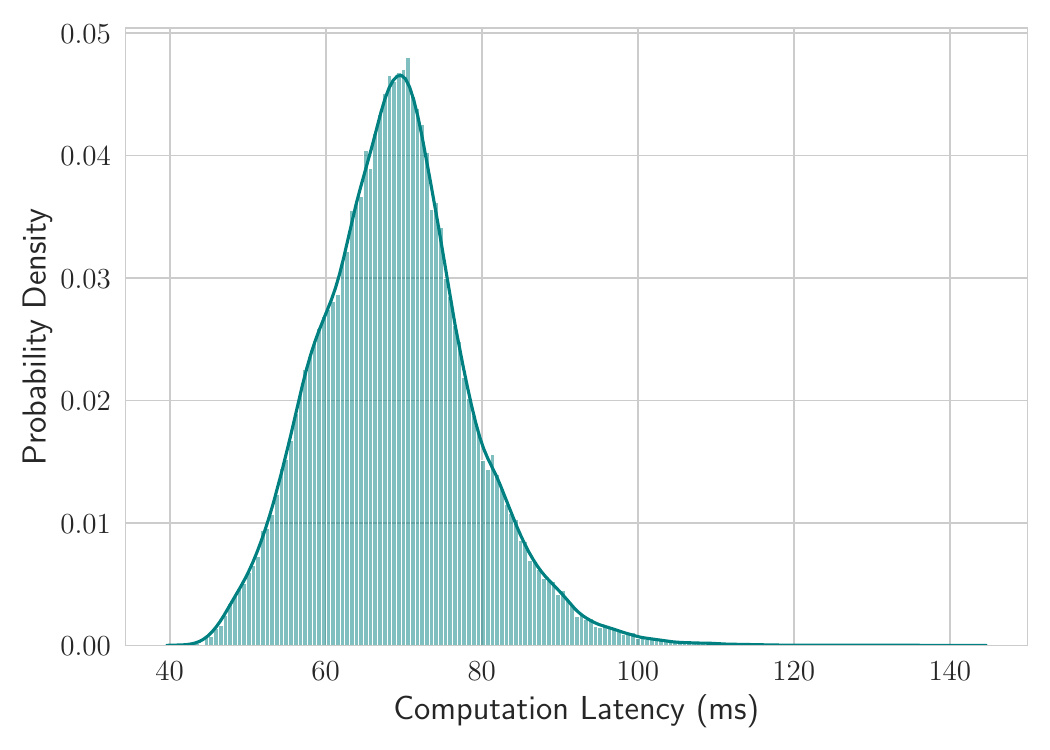}
    \caption{Computation time distribution of detection.}
    \label{fig:comp}
\end{figure}

Two polices based on video content are adopted for comparison. The first is NCC (Normalized Cross Correlation) policy  \cite{apicharttrisorn2019frugal}. NCC refers to the cross-correlation between two regions. A small cross-correlation value suggests that the detected object has significant change, and thus the tracker is likely to be inaccurate. Thus, NCC value is plugged into \eqref{index} as done in Max-Reduction policy. The second is CIB (Current IoU Based-) policy. In CIB policy, we assume the scheduler knows the IoU between the tracking position and the actual position. The IoU value is plugged into \eqref{index} for scheduling. Note that CIB policy requires knowledge of the actual position and thus cannot be implemented in real scenario. We just use it for comparison. 
The parameter $V$ is set to be $0.01$. 

\begin{figure}[htbp]
     \centering
     \begin{subfigure}{0.45\textwidth}
         \centering
         \includegraphics[width=\textwidth]{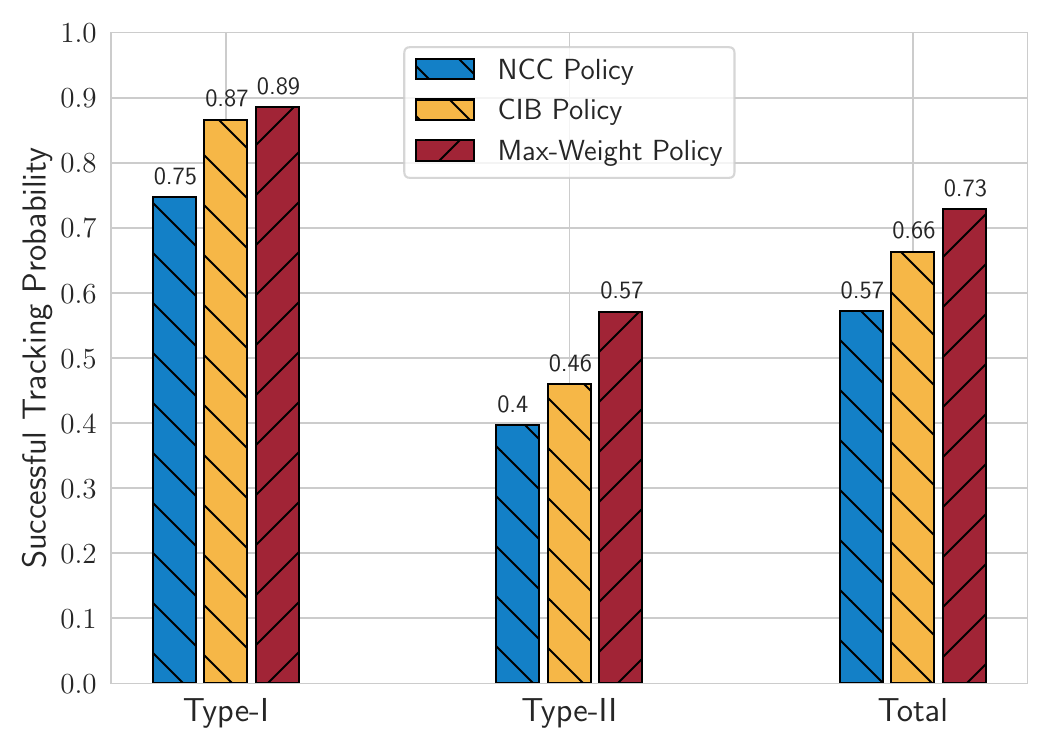}
         \caption{Successful tracking probability results}
         \label{fig:accuracy}
     \end{subfigure}
     \hfill
     \begin{subfigure}{0.45\textwidth}
         \centering
         \includegraphics[width=\textwidth]{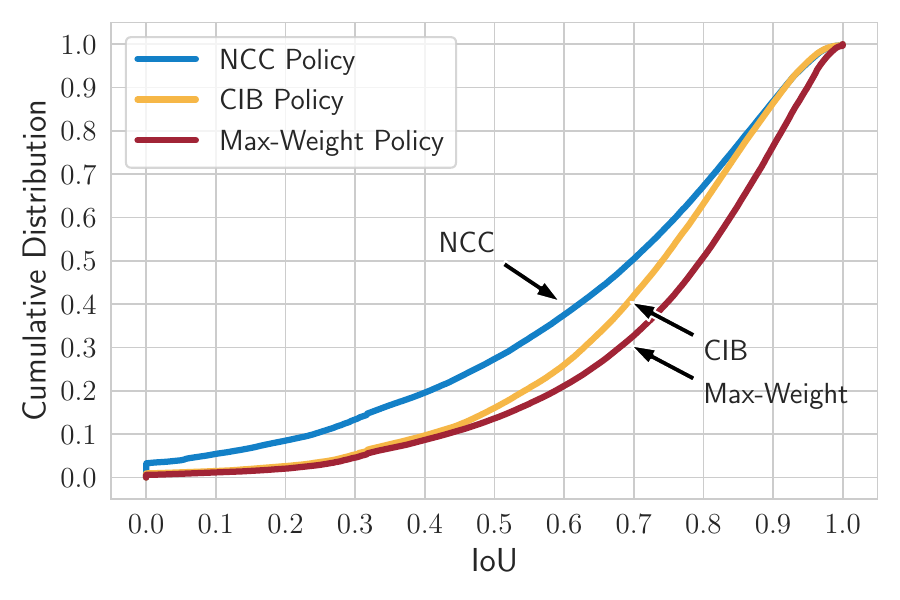}
         \caption{Cumulative distribution of IoU}
         \label{fig:CDF}
     \end{subfigure}
     \hfill
        \caption{Performance comparison between two policies.}
        \label{fig:result}
\end{figure}

To evaluate the performance on ILSVRC17-VID dataset, we randomly take 300 videos from the videos that are not used for profiling. Fig. \ref{fig:accuracy} compares the success tracking probability of Max-Weight policy, NCC and CIB. As shown in this figure, Max-Weight policy outperforms the other two for both two types of devices, and improves the total successful tracking probability by 27\% compared with NCC. In Fig. \ref{fig:CDF}, the cumulative distributions of IoU under these two policies are presented. As we can see, Max-Weight policy enjoys better tracking performance. 

It is surprising to observe in Fig. \ref{fig:result} that CIB policy is worse than the Max-Weight policy, as CIB uses knowledge of the actual IoU. This phenomenon might be due to two reasons. First, CIB policy doesn't take transmission and computation delay into consideration, and this might lead to bad resource allocation. Second, CIB policy only uses current IoU, while the profiling curves in Fig. \ref{fig:fitting} incorporate long-term information. This motivates us to investigate how to represent content semantics from a time perspective.

\section{Conclusion}
\label{sec-5}
To support emerging real-time applications with computation-intensive status update, it is critical to efficiently manage communication and computation resources in the network system to provide as fresh status information as possible. To fully utilize computation resources, we considered a hybrid computation framework where computation tasks can be processed on-device or offloaded to an edge server. Task-specific timeliness requirement was modeled as penalty functions of AoI. We first analyzed the minimum average AoI penalty and formulated an optimization problem to compute the penalty lower bound. Based on the lower bound, we proposed indices to quantify the priorities of local computing and edge computing respectively. Combining energy virtual queue with these indices, we proposed a Max-Weight scheduling policy, inspired by the optimal conditions of the lower bound problem. Extensive simulations showed that our proposed policy has close-to-optimal performance under different penalty functions. We also applied the proposed policy to object tracking tasks on ILSVRC17-VID dataset and improved the tracking accuracy compared with scheduling polices based on video content information.
\appendices

\section{Proof of Lemma \ref{thm_0}, \ref{thm_0_1}}
\label{app-0}
Given a policy $\pi$, let $C(k)$ be the number of update rounds finished by time slot $k$, and $Z_i$ be the time slot when the $i$-th round starts. The average AoI penalty by time slot $K$ is
\begin{equation}
	\label{eq1-1}
	\begin{aligned}
    &\frac{1}{K}\left(\sum_{k=1}^{K}f(h(k))\right) \\
    &=
    \frac{1}{K}
        \left(
            \sum_{c=1}^{C(K)}(F(h^{+}_c)-F(h^{-}_c-1))+R(K)
        \right),
\end{aligned}
\end{equation}

where $h_c^+$ is the peak age of the $c$-th update round, and $h_c^{-}$ be the age at the beginning of this round, and
\begin{align}
    R(K)\triangleq \sum_{h=h^{-}_{C(K)+1}}^{K+h^{-}_{C(K)+1}-Z_{C(K)+1}}f(h).
\end{align}

Update rounds can be further classified based on the type of computation executed during the round. Let $C_{l}(K)$ be the number of rounds that contains local computing, and $C_{t}(K)$ be the number of rounds with edge computing, then
\begin{align}
    \frac{1}{K}\sum_{c=1}^{C(K)}F(h^{+}_c) = \frac{1}{K}\sum_{c=1}^{C_{l}(K)}F(h^{+}_{l,c}) + \frac{1}{K}\sum_{c=1}^{C_{t}(K)}F(h^{+}_{t,c}),
\end{align}
where $h^{+}_{l,c}$ is the peak age in the $c$-th local computing round, and $h^{+}_{t,c}$ is the peak age in the $c$-th edge computing round. As for $h^{-}_c$, note that it equals the total latency in round $c-1$. Thus we can shift the summation by one round and obtain
\begin{equation}
	\begin{aligned}
    &\frac{1}{K}\sum_{c=2}^{C(K)+1}F(h^{-}_c-1)\\
    &=\frac{1}{K}\sum_{c=1}^{C_{l}(K)}F(D_{l,c}-1) + \frac{1}{K}\sum_{c=1}^{C_{t}(K)}F(D_{t,c}+D_{e,c}-1).
	\end{aligned}
\end{equation}

Due to the independence of communication and computation latency in each update round, basic renewal theory yields the following equations:
\begin{equation}
	\begin{aligned}
    &\lim_{K\to\infty} \frac{C_{l}(K)}{K} =  \frac{\rho_{l}(\pi)\overline{E}}{E_{l}\overline{D}_{l}}, ~\text{w.p.1},\\
    &\lim_{K\to\infty} \frac{C_{t}(K)}{K} =  \frac{\rho_{t}(\pi)\overline{E}}{E_{t}\overline{D}_{t}}, ~\text{w.p.1}.
	\end{aligned}
\end{equation}

Obviously, policies with unbounded $R(K)$ as $K\to\infty$ cannot be optimal. Therefore, we only consider policies under which the residual term $R(K)$ satisfies
\begin{align}
    \lim_{K\to\infty}\frac{R(K)}{K} = 0,~\text{w.p.1.} 
\end{align}

Since these time averages converge with probability 1, the Lebesgue Dominated Convergence Theorem \cite{billingsley2013convergence} ensures the time average expectations are the same as the pure time averages. Thus, by letting $K\to\infty$ in \eqref{eq1-1}, the average AoI penalty under policy $\pi$ can be written as 
\begin{equation}
\begin{aligned}
        &\frac{\rho_{l}(\pi)\overline{E}}{E_{l}\overline{D}_{l}}\mathbb{E}_{\pi}[F(h^{+}_{l})-F(D_{l}-1)]\\
        & + 
        \frac{\rho_{t}(\pi)\overline{E}}{E_{t}\overline{D}_{t}}\mathbb{E}_{\pi}[F(h^{+}_{t})-F(D_{t}+D_{e}-1)].
\end{aligned}
\end{equation}

This concludes the proof for Lemma \ref{thm_0}.

On the other hand, based on the definition of $C(k)$, the following two inequalities hold
\begin{align}
    \sum_{c=1}^{C(K)}(h^{+}_c - h^{-}_c+1) \le K,\\
    \sum_{c=1}^{C(K)+1}(h^{+}_c - h^{-}_c+1) \ge K.
\end{align}
Let $K'$ be the time slot when round $C(K)+1$ finishes, then
\begin{align}
    \lim_{K\to\infty}\frac{K'}{K} = 1.
\end{align}
Thus, 
\begin{equation}
\begin{aligned}
	&\lim_{K\to\infty}\frac{1}{K}\sum_{c=1}^{C(K)}(h^{+}_c - h^{-}_c+1) \le 1,\\
	&1 \le \lim_{K\to\infty}\frac{K'}{K}\frac{1}{K'}\sum_{c=1}^{C(K)+1}(h^{+}_c - h^{-}_c+1).
\end{aligned}
\end{equation}
Applying renewal theory again yields
\begin{equation}\label{constraint_2}
	\begin{aligned}
		&\frac{\rho_{t}(\pi)\overline{E}}{E_{t}\overline{D}_{t}}(\mathbb{E}_{\pi}[h^{+}_{t}]-(\overline{D}_{t}+\overline{D}_{e}-1))\\
		&+\frac{\rho_{l}(\pi)\overline{E}}{E_{l}\overline{D}_{l}}(\mathbb{E}_{\pi}[h^{+}_{l}]-(\overline{D}_{l}-1)) = 1.
	\end{aligned}
\end{equation}
Once the constraints in \textbf{P1} are satisfied, so does \eqref{constraint_2}. This concludes the proof for Lemma \ref{thm_0_1}.

\section{Proof of Lemma \ref{thm_1}}
\label{app-1}
To show that \textbf{P4} is a convex problem, we only need to prove that the following is a convex function of $\rho_{l,n}$ and $\rho_{t,n}$,
\begin{align}\label{eq-app-2-1}
	\left(\frac{\rho_{l,n}\overline{E}_n}{E_{l,n}\overline{D}_{l,n}}+\frac{\rho_{t,n}\overline{E}_n}{E_{t,n}\overline{D}_{t,n}}\right)\tilde{F}_n(G_n(\rho_{l,n},\rho_{t,n})).
\end{align}
For simplicity, let's introduce some notations:
\begin{equation}
	\begin{aligned}
		&p\triangleq \overline{D}_{l,n}-1,~q\triangleq \overline{D}_{t,n}+\overline{D}_{e,n}-1, \\
		&x \triangleq \frac{\rho_{l,n}\overline{E}_n}{E_{l,n}\overline{D}_{l,n}},~ y\triangleq \frac{\rho_{t,n}\overline{E}_n}{E_{t,n}\overline{D}_{t,n}}.
	\end{aligned}
\end{equation}

According to \eqref{upper-opt-single}, the function $G_n(\rho_{l,n},\rho_{c,n})$ can be expressed as 
\begin{align}
	G_n(\rho_{l,n},\rho_{c,n}) = \frac{1+px+qy}{x+y}, 
\end{align}
and the target \eqref{eq-app-2-1} becomes
\begin{align}
	(x+y)\tilde{F}_n\left(\frac{1+px+qy}{x+y}\right).
\end{align}

Let $Z(x,y)$ be the perspective transform of $\tilde{F}_n$
\begin{align}
	Z(x,y)\triangleq y\tilde{F}_n\left(\frac{x}{y}\right),~y>0.
\end{align}

Because $\tilde{F}_n(x)$ is convex, its perspective transform $Z(x,y)$ is also convex \cite{boyd2004convex}. Consider the following linear transformation
\begin{align}
	\begin{bmatrix}
		1+px+qy \\ x+y
	\end{bmatrix} = 
	\begin{bmatrix}
		p & q \\
		1 & 1
	\end{bmatrix}
	\begin{bmatrix}
		x \\ y
	\end{bmatrix} + 
	\begin{bmatrix}
		1 \\ 0
	\end{bmatrix},
\end{align}
let $\bm{x}\triangleq (x,y)^T$, and denote the relationship above as $A\bm{x}+\bm{b}$. It is easy to see that 
\begin{align}
	(x+y)\tilde{F}_n\left(\frac{1+px+qy}{x+y}\right) = Z(A\bm{x}+\bm{b}).
\end{align}
According to the composition rule of convex function, the target \eqref{eq-app-2-1} is also a convex function. Thus, \textbf{P4} is a convex optimization problem.

\section{Proof of Theorem \ref{thm-lagrange}}
\label{app-1-2}
According to KKT conditions, the optimal solution satisfies
\begin{align*}
    & a_n \tilde{F}_n(G_n(x_n^*,y_n^*)) + 
    \frac{a_n b_n y_n^* (c_n-d_n) - a_n}{a_n x_n^* + b_n y_n^*} \tilde{f}_n (G_n(x_n^*,y_n^*))\\
    &- a_n v_n + \beta_n^* - \gamma_n^* = 0,~\forall n\in\mathcal{N}, \\
    & b_n \tilde{F}_n(G_n(x_n^*,y_n^*)) + 
    \frac{a_n b_n x_n^* (d_n-c_n) - b_n}{a_n x_n^* + b_n y_n^*} \tilde{f}_n (G_n(x_n^*,y_n^*))\\
    &- b_n w_n + \frac{\overline{E}_n}{E_{t,n}}\alpha^* + \beta_n^* - \nu_n^* = 0,~\forall n\in\mathcal{N}, \\
    & \alpha \left( \sum_{n\in\mathcal{N}} \frac{y_n^*\overline{E}_n}{E_{t,n}}-M \right) = 0, 
    ~ \beta_n^* (x_n^*+y_n^*-1) = 0,\\
    &\gamma_n^* x_n^* = 0,
    ~ \nu_n^* y_n^* =0,~\forall n\in\mathcal{N}.
\end{align*}
Combining the first two equations and removing $\beta_n^*$ gives 
\begin{equation}\label{KKT-1}
    \begin{aligned}
    &\frac{a_n b_n(d_n-c_n)(x_n^*+y_n^*) + a_n - b_n}{a_n x_n^* + b_n y_n^*} \tilde{f}_n (G_n(x_n^*,y_n^*))\\
    &+(b_n-a_n)\tilde{F}_n(G_n(x_n^*,y_n^*))\\
    &+a_n v_n + \gamma_n^* - b_n w_n - \nu_n^* = -\frac{\overline{E}_n}{E_{t,n}}\alpha^*.
\end{aligned}
\end{equation}

Recalling the definition of $G_n$, it can be written as
\begin{align}
    G_n(x_n,y_n) = \frac{1+a_n c_n x_n+b_n d_n y_n }{a_n x_n + b_n y_n}.
\end{align}
With direct computation, we have
\begin{align}\label{eq1-2}
    &\frac{a_n b_n(d_n-c_n)(x_n+y_n) + a_n - b_n}{a_n x_n + b_n y_n} \\
    &= (a_n-b_n)G_n(x_n,y_n) + b_n d_n - a_n c_n.
\end{align}
As shown in \eqref{upper-opt-single}, $G_n(\rho_{l,n},\rho_{t,n})$ is the expected peak age when the energy allocation is $\rho_{l,n},\rho_{t,n}$. Combining this fact with \eqref{eq1-2}, we can express \eqref{KKT-1} as  
\begin{equation}\label{KKT-final}
    \frac{W_{t,n}(h_{t,n})}{E_{t,n}\overline{D}_{t,n}}  -  \frac{W_{l,n}(h_{l,n})}{E_{l,n}\overline{D}_{l,n}}  - \frac{\gamma_n^* - \nu_n^*}{\overline{E}_n} = \frac{\alpha^*}{E_{t,n}}.
\end{equation}
This concludes the proof.

\section{Proof of Theorem \ref{thm-MW}}
\label{app-2}
    First, to maximize \eqref{eq-weight}, the scheduling policy should only consider devices in the set $\mathcal{C}_l\cup \mathcal{C}_t$. We start a simple policy under which 1) all devices in $\mathcal{C}_l$ are scehduled to do local computing, 2) devices in $\mathcal{C}_t - \mathcal{C}_l$ are sort in descending order according to the value of $I_n(k)$, and at most $m(k)$ devices from the top are scheduled to offload.

    Consider two cases:
    \begin{enumerate}
        \item If there are less than $m(k)$ devices in $\mathcal{C}_t - \mathcal{C}_l$, we can reorder devices from $\mathcal{C}_l$ to offload if $I_n(k)\ge 0$ until the channels are all occupied.
        \item If all channels are occupied by devices from $\mathcal{C}_t - \mathcal{C}_l$, taking the device with the largest $I_n(k)$ in $\mathcal{C}_l\cap \mathcal{C}_t$, say device $x$. If the index of device $x$ is larger than one of the devices scheduled to offload, say device $y$, then we can replace $y$ by $x$ to offload, and improve the sum weight \eqref{eq-weight}. Repeating this process finite times will maximize \eqref{eq-weight}. 
    \end{enumerate}

    The process above is equivalent to first sort devices in $\mathcal{C}_t$ in descending by the value of $I_n(k)$, then order at most $m(k)$ devices with $I_n(k)\ge 0$ to offload, corresponding to Line 1-10 in Alg.\ref{Alg-0}. Devices left in $\mathcal{C}_l$ will be scheduled to do local computing, corresponding to Line 11-13 in Alg.\ref{Alg-0}.

\section{Proof of Proposition \ref{prop-2}}
\label{app-prop-2}
The proof is divided into three parts. In the first part, the weight in \eqref{eq-weight} is derived by computing the drift of an expectation term. Next, we show how to obtain a randomized policy $\pi_{\text{R}}$. Finally, we prove \eqref{performance} by comparing $\pi_{\text{MW}}$ with $\pi_{\text{R}}$.

\subsection{Drift Expression}

We first consider the drift of the quadratic virtual queue functions, defined as $\Delta(k)$:
\begin{align}
    \Delta(k) &\triangleq \frac{1}{2}\mathbb{E}\left[\left.\sum_{n\in\mathcal{N}} Q_n^2(k+1)-\sum_{n\in\mathcal{N}} Q_n^2(k)~\right|~\mathcal{H}(k)\right],
\end{align}
recalling that $\mathcal{H}(k)$ represents the history up to time slot $k$. $\Delta(k)$ satisfies that  
\begin{equation}\label{drift-1}
	\begin{aligned}
	    \Delta(k) 
	    \le& \frac{1}{2}\mathbb{E}\left[\left. \sum_{n\in\mathcal{N}}(E_n(k)-\overline{E}_n)^2 \right| \mathcal{H}(k)\right] \\
	    &+ \mathbb{E}\left[\left. \sum_{n\in\mathcal{N}}Q_n(k)(E_n(k)-\overline{E}_n) \right| \mathcal{H}(k)\right]\\
	    \le& \frac{1}{2}\sum_{n\in\mathcal{N}}(\max(E_{l,n},E_{t,n})-\overline{E}_n)^2 \\
	    &+ \mathbb{E}\left[\left. \sum_{n\in\mathcal{N}}Q_n(k)(E_n(k)-\overline{E}_n) \right| \mathcal{H}(k)\right].
	\end{aligned}
\end{equation}

As for the age part, it is captured by the following function
\begin{align}
    P(k) \triangleq \sum_{n\in\mathcal{N}} ((h_n(k)-1)f_n(h_n(k)-1) - \tilde{F}_n(h_n(k)-1)).
\end{align}

When $D_{l,n} = D_{t,n} = 1, D_{e,n} = 0,~\forall n\in\mathcal{N}$, the indices in \eqref{eq1-3} and \eqref{eq1-4} are reduced to be 
\begin{align}
	W_n(h_n(k)) \triangleq h_n(k)f_n(h_n(k)) - \tilde{F}_n(h_n(k)).
\end{align}
The drift is 
\begin{equation}
	\begin{aligned}
		\Gamma(k) \triangleq& \mathbb{E}\left[P(k+1)-P(k)|\mathcal{H}(k)\right] \\
	    =& \sum_{n\in\mathcal{N}} (h_n(k)-1)(f_n(h_n(k)) - f_n(h_n(k)-1)\\
	    &+ \sum_{n\in\mathcal{N}}(f_n(h_n(k))+\tilde{F}_n(h_n(k)-1) - \tilde{F}_n(h_n(k))) \\
	    &- \sum_{n\in\mathcal{N}}W_n(h_n(k))\mathbb{E}\left[(u_{t,n}(k)+u_{l,n}(k))|\mathcal{H}(k)\right]. 
	\end{aligned}
\end{equation}  

Because
\begin{align}
	&\tilde{F}_n(h_n(k)-1) - \tilde{F}_n(h_n(k)) \le -f_n(h_n(k)-1),
\end{align}
the drift term $\Gamma(k)$ satisfies
\begin{equation} \label{drift-2}
	\begin{aligned}
		\Gamma(k) \le& \sum_{n\in\mathcal{N}} h_n(k)(f_n(h_n(k)) - f_n(h_n(k)-1))\\
		&- \sum_{n\in\mathcal{N}}W_n(h_n(k))\mathbb{E}\left[(u_{t,n}(k)+u_{l,n}(k))|\mathcal{H}(k)\right].
	\end{aligned}
\end{equation}

Let $B \triangleq \frac{V}{2}\sum_{n\in\mathcal{N}}(\max(E_{l,n},E_{t,n})-\overline{E}_n)^2$, combining \eqref{drift-1} and \eqref{drift-2} yields
\begin{equation}\label{drift_sum}
	\begin{aligned}
		&\Gamma(k) + V \Delta(k) \\
		&\le B - V\sum_{n\in\mathcal{N}}Q_n(k)\overline{E}_n \\
		&+ \sum_{n\in\mathcal{N}} h_n(k)(f_n(h_n(k)) - f_n(h_n(k)-1)  \\
		&- \sum_{n\in\mathcal{N}} (W_n(h_n(k))-VE_{l,n}Q_n(k)) \mathbb{E}\left[u_{l,n}(k)|\mathcal{H}(k)\right] \\
		&- \sum_{n\in\mathcal{N}} (W_n(h_n(k))-VE_{t,n}Q_n(k)) \mathbb{E}\left[u_{t,n}(k)|\mathcal{H}(k)\right].
	\end{aligned}
\end{equation}
Based on Theorem \ref{thm-MW}, the policy $\pi_{\text{MW}}$ makes decisions to minimize the right-hand-side of \eqref{drift_sum} at each time slot.

\subsection{Randomized Policy}

In this part, we show how to construct a randomized policy $\pi_{\text{R}}$. Since there are $M$ channels, the candidate scheduling action set is defined as $\mathcal{A}$
\begin{align}
	\mathcal{A} \triangleq \left\{(\bm{u}_{l},\bm{u}_t) \left| \sum_{n\in\mathcal{N}}u_{t,n}\le M; u_{l,n}+u_{t,n}\le 1, \forall n\in\mathcal{N}\right. \right\}.
\end{align}

Let's define a randomized policy $\pi_{\text{R}}$ that takes action $a\in\mathcal{A}$ with probability $p_{\pi_{\text{R}}}(a)$ in each time slot. Since $\pi_{\text{R}}$ is independent of the history $\mathcal{H}(k)$, $\mathbb{E}_{\pi_{\text{R}}}[u_{t,n}(k)|\mathcal{H}(k)]$ and $\mathbb{E}_{\pi_{\text{R}}}[u_{l,n}(k)|\mathcal{H}(k)]$ are stationary, and we denote them as $p_{l,n}$ and $p_{t,n}$ respectively. Note that we cannot simply define a randomized policy and state that it would schedule device $n$ to do local computing with probability $p_{l,n}$, and to offload with probability $p_{t,n}$. Because of the channel constraint, such a vanilla policy might be infeasible.

Let $\mathcal{P}_{\text{R}} \triangleq \{(\bm{p}_{l},\bm{p}_t)\}$ be the set of probability distributions achievable by a randomized policy. The associated energy allocation scheme set $\mathcal{E}_{\text{R}}$ is defined as  
\begin{equation}
	\begin{aligned}
 	\mathcal{E}_{\text{R}} \triangleq 
 	&\left\{ 
 	(\bm{\rho}_l,\bm{\rho}_t)
	 	\left| 
		 	\rho_{l,n} = \frac{E_{l,n}\overline{D}_{l,n}p_{l,n}}{\overline{E}_n},\right.\right.\\
		 	&\left.\rho_{t,n} = \frac{E_{t,n}\overline{D}_{t,n}p_{t,n}}{\overline{E}_n}, 
		 	(\bm{p}_{l},\bm{p}_t) \in \mathcal{P}_{\text{R}}
 	\right\}.
	\end{aligned}
\end{equation}
Let $\mathcal{E}$ be the set of all possible energy allocation schemes under any stationary policy $\pi$. According to \cite[{Lemma~4.17}]{neely2010stochastic}, we have $\mathcal{E} = \mathcal{E}_{\text{R}}$. 

Now, let $\pi_{\text{opt}}$ be the optimal stationary scheduling policy, and its associated energy allocation vectors are $\bm{\rho}_l^*$ and $\bm{\rho}_t^*$, plugging $\bm{\rho}_l^*$ and $\bm{\rho}_t^*$ into to the optimization target of \eqref{upper-opt} yields a lower bound of the minimum AoI penalty. Let $\bm{p}_{l}^*$ and $\bm{p}_{t}^*$ be the corresponding scheduling probability\footnote{Note that we cannot directly using the solution of the lower bound problem \eqref{upper-opt} to construct randomized policy because  energy allocation scheme given by the solution may not be achievable.}. According to \eqref{upper-opt}, we obtain AoI penalty lower bound, 
\begin{align}\label{LB}
	J^{*} = \sum_{n\in\mathcal{N}}\frac{\alpha_n}{p+1} \left(\frac{1}{p_{l,n}^{*}+p_{t,n}^{*}}\right)^{p}. 
\end{align}

Because $\pi_{\text{MW}}$ minimizes the right-hand-side of \eqref{drift_sum}, we have 
\begin{equation}
	\begin{aligned}
		\Gamma(k) + V \Delta(k) \le& B + \sum_{n\in\mathcal{N}} h_n(k)(f_n(h_n(k)) - f_n(h_n(k)-1)) \\
		&- V\sum_{n\in\mathcal{N}}Q_n(k)(\overline{E}_n - p_{l,n}^*E_{l,n} - p_{t,n}^*E_{t,n})\\
		&- \sum_{n\in\mathcal{N}}(p_{l,n}^*+p_{t,n}^*)W_n(h_n(k)).
	\end{aligned}
\end{equation}

Because $p_{l,n}^*E_{l,n} + p_{t,n}^*E_{t,n} \le \overline{E}_n$, the inequality above is relaxed to be 
\begin{equation}\label{drift_sum_random}
	\begin{aligned}
		&\Gamma(k) + V \Delta(k) \\
		&\le B + \sum_{n\in\mathcal{N}} h_n(k)(f_n(h_n(k)) - f_n(h_n(k)-1)) \\
		&- \sum_{n\in\mathcal{N}}(p_{l,n}^*+p_{t,n}^*)W_n(h_n(k)).
	\end{aligned}
\end{equation}

\subsection{Performance Derivation}

In this part, we will prove the inequality \eqref{performance}. Let the average AoI penalty of device $n$ under policy $\pi_{\text{MW}}$ be $J_n^{\pi_{\text{MW}}}$. Note that the AoI penalty function for device $n$ is $f_n(h)=\alpha_n h^p$. In the $i$-th round of status update, $h_n(k)$, the AoI of device $n$, will increase from $1$ to a peak AoI $\hat{h}^{(i)}_n$. The average AoI penalty in this round is 
\begin{align}
	\alpha_n\sum_{h=1}^{\hat{h}^{(i)}_n} h^{p} \ge \alpha_n \int_{0}^{\hat{h}^{(i)}_n} h^{p} dh = \frac{\alpha_n}{p+1} (\hat{h}^{(i)}_n)^{p+1}.
\end{align}
Let $\hat{h}_n$ be the distribution of the peak AoI of device $n$, we have 
\begin{align}
	J_n^{\pi_{\text{MW}}} = \lim_{K\to\infty}\frac{1}{K}\mathbb{E}_{\pi_{\text{MW}}} \left[\sum_{k=1}^{K}h_n^p(k)\right] \ge 
	\frac{\alpha_n}{p+1}\frac{\mathbb{E}_{\pi_{\text{MW}}}[\hat{h}_n^{p+1}]}{\mathbb{E}_{\pi_{\text{MW}}}[\hat{h}_n-1]}.
\end{align}
As for the second term in \eqref{drift_sum_random}, we have 
\begin{equation}
	\begin{aligned}
	& \lim_{K\to\infty}\frac{1}{K}\mathbb{E}_{\pi_{\text{MW}}} \left[\sum_{k=1}^{K} h_n(k)(f_n(h_n(k)) - f_n(h_n(k)-1)) \right] \\
	& = \alpha_n\frac{\mathbb{E}_{\pi_{\text{MW}}}[\hat{h}_n^{p+1}]}{\mathbb{E}_{\pi_{\text{MW}}}[\hat{h}_n-1]} - J_n^{\pi_{\text{MW}}}\\
	& \le pJ_n^{\pi_{\text{MW}}}.
\end{aligned}
\end{equation}

Taking expectation under $\pi_{\text{MW}}$ on both sides of \eqref{drift_sum_random}, taking average up to time slot $K$ and letting $K$ to $\infty$, we have
\begin{equation}\label{drift_final}
\begin{aligned}
	&\lim_{K\to\infty} \frac{V}{2K} \mathbb{E}_{\pi_{\text{MW}}}\left[ \sum_{n\in\mathcal{N}}(Q_n^2(k+1) - Q_n^2(1))\right] \\
	&\le B + pJ^{\pi_{\text{MW}}}  - \lim_{K\to\infty} \frac{1}{K} \mathbb{E}_{\pi_{\text{MW}}}\left[P(k)-P(1)\right] \\
	&- \lim_{K\to\infty}\frac{1}{K}\mathbb{E}_{\pi_{\text{MW}}}\left[\sum_{n\in\mathcal{N}}\sum_{k=1}^{K}(p_{l,n}^*+p_{t,n}^*)W_n(h_n(k))\right].
\end{aligned}
\end{equation}

According to Theorem \ref{thm_4}, the left-hand-side of \eqref{drift_final} is  $0$ as $K\to\infty$, therefore,
\begin{equation}
\begin{aligned}
	&\lim_{K\to\infty}\frac{1}{K}\mathbb{E}_{\pi_{\text{MW}}}\left[\sum_{n\in\mathcal{N}}\sum_{k=1}^{K}(p_{l,n}^*+p_{t,n}^*)W_n(h_n(k))\right] \\
	&\le B + pJ_n^{\pi_{\text{MW}}}.
\end{aligned}
\end{equation}

Next, we study the property of the function $x\tilde{f}(x)-\tilde{F}(x)$. 
\begin{lemma}\label{lemma-1}
	Consider a differentiable injective function $\tilde{f}:\mathbb{R}\rightarrow \mathbb{R}$, its inverse is denoted as $\tilde{f}^{-1}$, and its integral is $\tilde{F}(x)$. If $\tilde{f}(x)$ is increasing, $y\tilde{f}^{-1}(y)-\tilde{F}(\tilde{f}^{-1}(y))$ is convex. 
\end{lemma}
\begin{proof}
	Let $S(y)\triangleq y\tilde{f}^{-1}(y)-\tilde{F}(\tilde{f}^{-1}(y))$. It's derivative is 
	\begin{align}
		\frac{\text{d} S(y)}{\text{d} y} = y \frac{\text{d} \tilde{f}^{-1}(y)}{\text{d} y} + \tilde{f}^{-1}(y) - y \frac{\text{d} \tilde{f}^{-1}(y)}{\text{d} y} = \tilde{f}^{-1}(y).
	\end{align}
	Because $\tilde{f}(x)$ is increasing, so is $\tilde{f}^{-1}(y)$. Therefore, the derivative of $S(y)$ is increasing, and thus $S(y)$ is convex.
\end{proof}

Because $\tilde{f}(x) = f(x)$ when $x\in\mathbb{N}$, we have $W_n(h_n(k)) = S_n(f_n(h_n(k)))$. With this equation, \eqref{drift_final} can be written as 
\begin{equation}\label{drift_final_2}
\begin{aligned}
	&B + pJ^{\pi_{\text{MW}}} \\
	&\ge \lim_{K\to\infty}\frac{1}{K}\mathbb{E}_{\pi_{\text{MW}}}\left[\sum_{n\in\mathcal{N}}\sum_{k=1}^{K}(p_{l,n}^*+p_{t,n}^*)S_n(f_n(h_n(k)))\right] \\
	&\overset{(a)}{\ge}  \sum_{n\in\mathcal{N}} (p_{l,n}^*+p_{t,n}^*)S_n\left(\lim_{K\to\infty}\frac{1}{K}\sum_{k=1}^{K}\mathbb{E}_{\pi_{\text{MW}}}\left[ h_n(k)\right]\right) \\
	&= \sum_{n\in\mathcal{N}} (p_{l,n}^*+p_{t,n}^*)S_n\left(J_n^{\pi_{\text{MW}}}\right),
\end{aligned}
\end{equation}
where inequality $(a)$ is due to Jenson's inequality.

Denoting $\tilde{f}_n^{-1}(J_n^{\pi_{\text{MW}}})$ by $\overline{h}_n$, we have
\begin{align}
	S_n\left(J_n^{\pi_{\text{MW}}}\right) = W_n(\overline{h}_n) = \frac{\alpha_n p}{p+1} \overline{h}_n^{p+1}.
\end{align}
Then, \eqref{drift_final_2} becomes
\begin{align}
	\frac{ p}{p+1} \sum_{n\in\mathcal{N}} \alpha_n(p_{l,n}^*+p_{t,n}^*)\overline{h}_n^{p+1} \le B + pJ^{\pi_{\text{MW}}}.
\end{align}
Multiplying both sides by $(J^*)^{\frac{1}{p+1}}$, the lower bound from \eqref{LB}, yields
\begin{equation}
	\begin{aligned}
	&(J^*)^{\frac{1}{p+1}}
		\left(\frac{p}{p+1}\sum_{n\in\mathcal{N}}\alpha_n(p_{l,n}^{*}+p_{t,n}^{*})\overline{h}_n^{p+1}\right)^{\frac{p}{p+1}} \\
		&\le 
		(J^*)^{\frac{1}{p+1}}\left(B+p J^{\pi_{\text{MW}}}\right)^{\frac{p}{p+1}}. 
	\end{aligned}
\end{equation}
Applying H\"older's inequality to the left hand side of the above gives
\begin{align}\label{drift_final_3}
	p^{\frac{p}{p+1}}(p+1)^{-1}\sum_{n\in\mathcal{N}}\alpha_n \overline{h}_n^p \le (J^*)^{\frac{1}{p+1}}\left(B + p J^{\pi_{\text{MW}}}\right)^{\frac{p}{p+1}}.
\end{align}
Because $\overline{h}_n = \tilde{f}_n^{-1}(J_n^{\pi_{\text{MW}}})$, rearranging \eqref{drift_final_3} yields
\begin{align}
    \left(\frac{J^{\pi_{\text{MW}}}}{p+1}\right)^{p+1} \le J^*\left(\frac{B}{p}+J^{\pi_{\text{MW}}}\right)^{p}.
\end{align}


\bibliographystyle{IEEEtran}
\end{document}